\documentclass[journal]{IEEEtran}
\usepackage{graphicx}
\graphicspath{{../2-Figures}}
\ifCLASSINFOpdf

\pdfoutput=1
\else

\fi

\usepackage{amsmath}
\usepackage{amssymb}
\usepackage{amsfonts}
\usepackage{amsthm}
\usepackage{algorithmic}
\usepackage{algorithm}
\usepackage{array}
\usepackage{stfloats}
\usepackage{color}
\usepackage{multirow}
\usepackage{makecell}
\usepackage[final]{changes}
\definecolor{customcolor}{HTML}{0072bd}
\setaddedmarkup{\textcolor{customcolor}{#1}}

\newtheorem{proposition}{Proposition}

\begin{document}
\bstctlcite{IEEEexample:BSTcontrol}

\title{High-Resolution Uplink Sensing in Millimeter-Wave ISAC Systems}
\author{Liangbin~Zhao,~\IEEEmembership{Student~Member,~IEEE,}
         Zhitong~Ni,~\IEEEmembership{Member,~IEEE,}
         Yimeng~Feng,~\IEEEmembership{Member,~IEEE,}
         Jianguo~Li,~\IEEEmembership{Member,~IEEE,}
         Xiangyuan~Bu,~\IEEEmembership{Member,~IEEE,}
         and~J.~Andrew~Zhang,~\IEEEmembership{Senior~Member,~IEEE}
\thanks{L. Zhao, J. Li, and X. Bu are with Beijing Institute of Technology, Beijing, 100081, China (emails: lb\_zhao\_bit\_ee@163.com, jianguoli@bit.edu.cn, bxy@bit.edu.cn).} 
 \thanks{Z. Ni and J. A. Zhang are with the University of Technology Sydney, NSW, 2007, Australia (emails: zhitong.ni@uts.edu.au, andrew.zhang@uts.edu.au).}
 \thanks{Y. Feng is with Macquarie University, NSW, 2109, Australia (email: yimeng.feng@mq.edu.au).}
}

\maketitle

\begin{abstract}
Perceptive mobile networks (PMNs), integrating ubiquitous sensing capabilities into mobile networks, represent an important application of integrated sensing and communication (ISAC) in 6G.
In this paper, we propose a practical framework for uplink sensing of angle-of-arrival (AoA), Doppler, and delay in millimeter-wave (mmWave) communication systems, which addresses challenges posed by clock asynchrony and hybrid arrays, while being compatible with existing communication protocols.
We first introduce a beam scanning method and a corresponding AoA estimation algorithm, which utilizes frequency smoothing to effectively estimate AoAs for both static and dynamic paths.
We then propose several methods for constructing a ``clean'' reference signal, which is subsequently used to cancel the effect caused by the clock asynchrony.
We further develop a signal ratio-based joint AoA-Doppler-delay estimator and propose an AoA-based 2D-FFT-MUSIC (AB2FM) algorithm that applies 2D-FFT operations on the signal subspace, which accelerates the computation process with low complexity.
Our proposed framework can estimate parameters in pairs, removing the complicated parameter association process. 
Simulation results validate the effectiveness of our proposed framework and demonstrate its robustness in both low and high signal-to-noise ratio (SNR) conditions.

\end{abstract}
 

\IEEEpeerreviewmaketitle

\section{Introduction}
\label{sec:introduction}
\subsection{Background and Motivations}
\label{subsec:background_and_motivations}
Integrated sensing and communication (ISAC) has emerged as an innovative technology that unifies sensing and communication functions within a single system, sharing hardware, spectrum resources, etc \cite{zhang2021overview}. 
This integration has the potential to greatly enhance efficiency and performance, attracting significant attention from academia and industry in recent years \cite{liu2020joint, zhang2022integration}. 
In the forthcoming 6G era \cite{ITUrecommendation2023framework}, millimeter-wave (mmWave) and terahertz (THz) frequency bands are promising candidates for ISAC systems.
Their shorter wavelengths enable smaller antennas, allowing for denser antenna arrays and significantly enhancing the potential for higher resolution in angle of arrival (AoA) estimation. 
The higher data rates and wider bandwidths also offer potential improvement in Doppler and delay estimation. 

Perceptive mobile networks (PMN), first proposed in 2017 \cite{zhang2017PMNframework}, enhance traditional communication-only mobile networks with ubiquitous sensing capabilities while maintaining uncompromising communication services. 
PMN enables three types of communication-centric sensing at base stations: uplink sensing, downlink active sensing, and downlink passive sensing. 
Among these options, uplink sensing -- where user equipment (UE) transmits signals that reflect or refract off environmental targets and are then received and processed by the base station (BS) for sensing \cite{rahman2019PMN} -- shows significant promise. 
This approach avoids the necessity for full-duplex operation and requires only minimal modifications to existing hardware and system architectures, making it a practical option for PMN deployment \cite{zhang2020PMN}. 

\subsection{Related Works}
\label{subsec:related_works}
\added{Achieving high-resolution and accurate estimations using uplink signals, especially in the mmWave or higher frequency bands, remains a challenging task.}
\added{Large-scale antenna array with hybrid architecture is commonly employed for mmWave communications \cite{sohrabi2016hybrid}.}
\added{In these hybrid architectures, the antenna array is divided into subarrays, each equipped with its analog-to-digital converter (ADC) or digital-to-analog converter (DAC).}
Compared to fully digital arrays, this architecture significantly reduces the number of ADCs/DACs, which introduces challenges in achieving high-resolution AoA estimation. 
Early works on AoA estimation in hybrid arrays, such as \cite{huang2010hybrid} and \cite{huang2011frequency}, primarily focused on estimating a single path. 
\added{Later, several beam scanning methods were proposed to reduce complexity, including sequential scanning \cite{singh2015feasibility} and hierarchical multi-resolution scanning \cite{noh2017multiresolutionCodebook}, albeit with low resolution.}
To improve resolution and support multipath estimation, algorithms based on compressive sensing \cite{lee2014exploitingCS} and multiple signal classification (MUSIC) have been proposed \cite{chuang2015high}. 
Nonetheless, accurately estimating coherent signals from static (without Doppler) paths originating from the same transmitter, and effectively separating static and dynamic (with Doppler) paths, remains challenging. 

\added{In addition to the difficulties in AoA estimation in hybrid arrays, uplink sensing typically employs a bi-static configuration, where the spatial separation between UE and BS leads to clock asynchrony, introducing unavoidable carrier frequency offset (CFO) and timing offset (TO) \cite{lu2024integrated}.}
These offsets cause ambiguities in Doppler and delay estimation, and hinder coherence sensing processing that relies on accurate channel phase information. 
\added{To address the issue of clock asynchrony in uplink sensing, two primary categories of methods have been proposed: offset cancellation methods and offset estimating-and-compensating methods \cite{wu2024sensing}.}
\added{Offset cancellation methods aim to directly mitigate the effects of CFO and TO by exploiting their consistent characteristics observed across different domains.}
\added{For instance, cross-antenna techniques leverage the fact that receiving antennas typically share a single local oscillator.}
Methods such as cross-antenna cross-correlation (CACC) \cite{li2017indotrack, qian2018widar2} and channel state information ratio (CSI-R) \cite{zeng2019farsense, li2022csi} use one antenna as a reference to reduce common phase differences effectively. 
However, these methods face limitations in low SNR environments and scenarios with multiple targets. 
Specifically, CACC can introduce mirror components through conjugate multiplication, while CSI-R may induce nonlinear distortion, both of which degrade overall performance. 
Although some studies have attempted to address these issues \cite{ni2021uplink-CACC, ni2023uplink-CSIR}, achieving high estimation accuracy remains challenging due to the reliance on approximation techniques such as filtering and Taylor expansion. 
\added{On the other hand, offset estimating-and-compensating methods adopt an analytical approach to estimate and compensate for CFO and TO during the sensing process.}
\added{Studies such as SpotFi \cite{kotaru2015spotfi} and Dyn-MUSIC \cite{li2016dynamic-MUSIC} have utilized subcarriers for joint AoA-delay estimation using the MUSIC algorithm.}
\added{These methods mitigate the effects of CFO and TO by applying linear phase adjustments to the channel state information (CSI), yet residual errors persist, and performance degrades in low SNR.}
\added{Other notable approaches, such as JUMP \cite{pegoraro2024jump} and the multiple target Doppler estimation (MTDE) algorithm utilizing nullspace techniques \cite{zhao2023multiple}, offer potential improvements in both accuracy and robustness, but still leave room for further enhancement.}
\added{Table \ref{tab:comparison} provides a comparison of related sensing techniques, highlighting the capabilities and performance differences between the proposed method and these techniques.}


\begin{table*}[t]
    \renewcommand{\arraystretch}{1.4}  
    \centering
    \scriptsize  
    \caption{Comparison of related sensing techniques. }
    \resizebox{\textwidth}{!}{
        \begin{tabular}{|c|c|c|c|c|c|c|c|c|c|}
            \hline
            \multirow{2}{*}{\textbf{Literature}} & \multicolumn{2}{c|}{\textbf{System Configuration}} & \multicolumn{3}{c|}{\textbf{Parameters}} & \multicolumn{4}{c|}{\textbf{Capabilities}} \\ \cline{2-10} 
            & \textbf{Bi-static} & \textbf{Array} & \textbf{AoA} & \textbf{Delay} & \textbf{Doppler} & \textbf{Multi-Targets}$^\dagger$ & \textbf{Joint Estimation} & \textbf{CFO/TO Elimination} & \textbf{Frame}$\ddagger$ \\ \hline
            \textbf{H-MUSIC}\cite{chuang2015high} & No & Hybrid & Yes & No & No & Multi-Static & No & No & Pilot \\ \hline
            \textbf{Indotrack}\cite{li2017indotrack} & Yes & Digital & Yes & No & Yes & No & No & CACC & Pilot \\ \hline
            \textbf{Widar2.0}\cite{qian2018widar2} & Yes & Digital & Yes & Yes & Yes & 1 Static \& Multi-Dyn. & AoA-Doppler-Delay (EM) & CACC & Pilot \\ \hline
            \textbf{Farsense}\cite{zeng2019farsense} & Yes & Digital & No & Yes & No & No & No & CSI-R & Any \\ \hline
            \textbf{M-CACC}\cite{ni2021uplink-CACC} & Yes & Digital & Yes & Yes & Yes & 1 Static \& Multi-Dyn. & Doppler-Delay (Mirror-MUSIC) & CACC & Pilot \\ \hline
            \textbf{T-CSI-R}\cite{ni2023uplink-CSIR} & Yes & Digital & Yes & Yes & Yes & Multi-Dyn. & AoA-Delay (MUSIC) & CSI-R & Any \\ \hline
            \textbf{SpotFi}\cite{kotaru2015spotfi} & Yes & Digital & Yes & Yes & No & Multi-Static & AoA-Delay (2D-MUSIC) & Linear Fit & Pilot \\ \hline
            \textbf{Dyn-MUSIC}\cite{li2016dynamic-MUSIC} & Yes & Digital & Yes & Yes & No & 1 Static \& Multi-Dyn. & AoA-Delay (2D-MUSIC) & Linear Fit & Pilot \\ \hline
            \textbf{JUMP}\cite{pegoraro2024jump} & Yes & Digital & Yes & Yes & Yes & 1 Static \& Multi-Dyn. & AoA-Delay (Cross-Correlation) & Time-Domain Alignment & Pilot \\ \hline
            \textbf{MTDE}\cite{zhao2023multiple} & Yes & Digital & Yes & No & Yes & 1 Static \& Multi-Dyn. & AoA-Doppler (MUSIC) & Spatial Filter & Pilot \\ \hline
            \textbf{Proposed}\rule[-9pt]{0pt}{22pt} & Yes & Hybrid & Yes & Yes & Yes & \makecell{Multi-Static \& Multi-Dyn. \\ (Beam Scanning + Smoothing)} & AoA-Doppler-Delay (AB2FM) & Signal Ratio & Any \\ \hline
        \end{tabular}
    }
    \vspace{-2pt}
    \label{tab:comparison}

    \begin{itemize}
        \item[$\dagger$]: ``Static'' and ``Dyn.'' represent static and dynamic targets. E.g., ``Multi-Static \& Multi-Dyn.'' indicates the ability to estimate multiple static and dynamic paths simultaneously.
        \item[$\ddagger$]: ``Pilot'' indicates the use of pilot-based frames, while ``Any'' allows any frame type.
    \end{itemize}
    \hrule
\end{table*}

\subsection{Contributions and Paper Organization}
\label{subsec:contributions_and_paper_organization}
Targeting to address these challenges, we propose a high-resolution uplink sensing framework for estimating AoA, Doppler, and delay in a bi-static mmWave system with hybrid array. 
This framework addresses several key challenges: (1) mitigating CFO and TO between the transmitter and receiver, (2) overcoming the limitations imposed by a restricted number of RF chains in hybrid array, (3) simplifying the estimation algorithm to reduce complexity, (4) enabling the use of data frames for sensing to enhance resource efficiency and maintain compatibility with existing protocols, and (5) enabling direct simultaneous estimation of parameters, thus avoiding the need for separate parameter estimation and subsequent complicated parameter matching. 
The primary contributions of this paper are summarized as follows.
\begin{itemize}
    \item We present a practical and robust framework for bi-static ISAC mmWave systems based on OFDM signals, which consists of two primary stages: the \textit{AoA estimation stage} (AES) and the \textit{Doppler-delay estimation plus communication stage} (DDE\&CS). 
    In this framework, pilot sequences are required only during AES for AoA estimation.
    Doppler and delay parameters can be extracted during the DDE\&CS using the data frames, eliminating the need for specific pilot sequences and ensuring compatibility with existing protocols.
    This design supports hybrid array with at least two subarrays and is tailored for bi-static sensing configurations.

    \item We propose a reliable method to mitigate CFO and TO using a constructed reference signal. 
    \added{In order to construct this signal, we begin by introducing a high-resolution AoA estimation algorithm that applies frequency smoothing to the MUSIC algorithm in AES, facilitating the simultaneous estimation of both static and dynamic paths.}
    \added{Following this, we develop various beamforming (BF) methods tailored to specific requirements based on the estimated AoAs,} enhancing line-of-sight (LoS) signals for communication while suppressing non-line-of-sight (NLoS) signals to offer a ``clean'' reference for sensing. 
    Lastly, we propose a signal ratio-based method for mitigating CFO and TO using the reference signal, enabling unambiguous estimation of Doppler and delay parameters in DDE\&CS.

    \item Utilizing the CFO-and-TO-free signal, we present the \textit{AoA-Based 2D-FFT-MUSIC} (AB2FM) algorithm that achieves low-complexity and high-resolution joint Doppler-delay estimation based on the estimated AoAs. 
    Utilizing the signal subspace and 2D-FFT, this approach accelerates MUSIC spectrum searches and inherently pairs estimated parameters, eliminating the need for separate complicated parameter matching.
    Leveraging the estimated Dopplers, we can effectively distinguish between static and dynamic paths, as the initial independent AoA estimation provides AoAs for both types of paths without differentiation.
\end{itemize}

The remainder of this paper is organized as follows. 
Section \ref{sec:system_model} describes the application scenario and system model. 
Section \ref{sec:AoA_estimation} details the frequency smoothed AoA estimation algorithm and the LoS-path detection scheme. 
In Section \ref{sec:estimation_of_doppler_and_delay}, we present communication beam design methods, CFO and TO elimination techniques, and derive a low-complexity MUSIC-based algorithm for Doppler-delay estimation.
Section \ref{sec:simulation_results} and Section \ref{sec:conclusion} present simulation results and conclusions, respectively.

\textit{Notation: }Bold uppercase and lowercase letters denote matrices and vectors (e.g., $\mathbf{A}$ and $\mathbf{a}$).
$\mathbf{1}_{m, n}$ represents an $m\times n$ matrix with all elements equal to one.
$\left(\mathbf{A}\right)_{i, j}$ refers to the $(i, j)$-th element of matrix $\mathbf{A}$, $\left(\mathbf{a}\right)_i$ indicates the $i$-th element of vector $\mathbf{a}$.
Superscript $(\cdot)^H$, $(\cdot)^T$, and $(\cdot)^*$ denote the Hermitian transpose, transpose and complex conjugate, respectively.
$\otimes$ and $\odot$ represent the Kronecker and Hadamard products.
$\operatorname{vec}\{\cdot\}$ and $\mathbb{E}\{\cdot\}$ denotes the matrix vectorization and the expectation operator.
$\operatorname{diag}(\cdot)$ denotes the diagonal matrix constructed from a vector. 
$|a|$, $\|\mathbf{a} \|_2$ and $\|\mathbf{A} \|_F$ denote the absolute value of scalar $a$, the Euclidean norm of vector $\mathbf{a}$ and the Frobenius norm of matrix $\mathbf{A}$, respectively.
$\Theta_{\setminus {\theta}}$ denotes the set $\Theta$ with the element $\theta$ removed.
All indices in this paper start from $0$.

\section{System Model}
\label{sec:system_model}
\subsection{System Setup and Signal Models}
\label{subsec:scenario_and_received_signal}
\begin{figure}[t]
    \centering
    \includegraphics[width=0.8\linewidth]{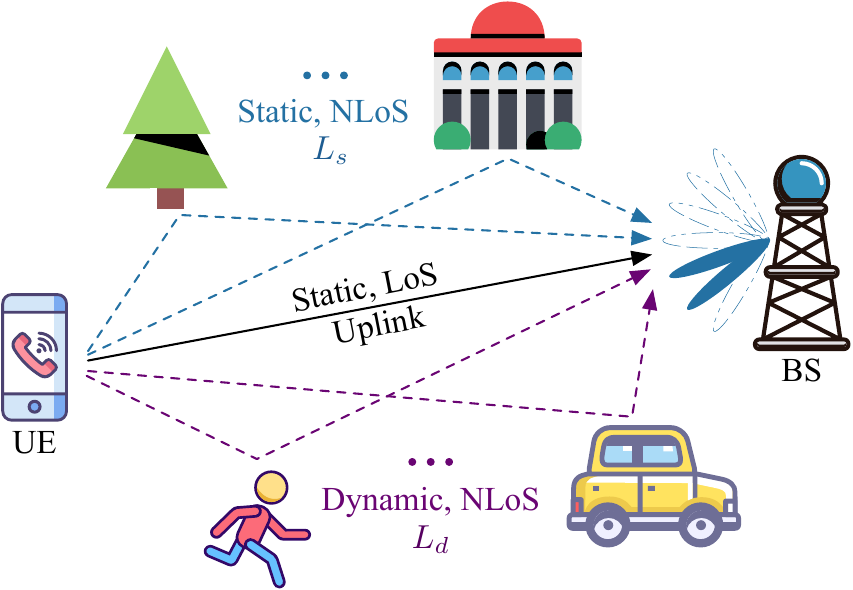}
    \caption{Illustration of the uplink sensing scenario, with $L_s$ static paths and $L_d$ dynamic paths.}
    \label{fig:Scenario}
\end{figure}

We consider an uplink ISAC scenario for a mmWave PMN, as illustrated in Fig. \ref{fig:Scenario}. 
The BS is equipped with a fully-connected hybrid array consisting of $N_{RF}$ analog subarrays, with a half-wavelength spaced uniform linear array (ULA) of $N_r$ antennas. 
For simplicity, we assume a ULA with $N_{RF}=2$ subarrays, namely, subarray $\mathcal{A}$ and subarray $\mathcal{B}$, as shown in the blue part of Fig. \ref{fig:system_structure}. 
\added{The uplink transmission is established with a static UE, which is equipped with an omnidirectional antenna, ensuring uniform signal radiation irrespective of direction.}
\added{Given the fixed BS-UE geometry and OFDM's widespread adoption in communication systems due to its robustness and spectral efficiency, as well as its potential for sensing applications \cite{zhang2021overview}, we employ OFDM modulation in our system configuration.}

Let $B$ and $K$ denote the signal bandwidth and the number of subcarriers, respectively. 
The subcarrier interval is thus $\Delta f = B/K$. 
The OFDM symbol period is $T_{s} = 1/\Delta f + T_{p}$ where $T_{p}$ is the period of cyclic prefix (CP). 
The transmitted signal for the $k$-th ($k=0, \cdots, K-1$) subcarrier is denoted as $x_k(t)$. 
For the BS, the received signal for the $k$-th subcarrier can be represented as \cite{zhang2021overview}
\begin{equation}
    \mathbf{r}_k(t) = \mathbf{h}_k(t)x_k(t)+ \mathbf{n}_k(t),
    \label{equ:received_signal_after_antennas}
\end{equation}
where $\mathbf{h}_k(t)$ is the CSI vector and $\mathbf{n}_k\sim \mathcal{C N}\left(\mathbf{0}, \sigma^2 \mathbf{I}_{N_r}\right)$ is the received additive white Gaussian noise (AWGN) at the BS. 

In a bi-static configuration, the transmitter and receiver are physically separated, leading to clock asynchrony. 
This asynchrony causes inevitable time-varying CFO, $f_o(t)$, and TO, $\tau_o(t)$, which are consistent among all paths. 
Let $L$ denote the total number of signal paths, and $L=L_s+L_d$, with $L_s$ and $L_d$ representing the number of static and dynamic paths, respectively.
\added{In mmWave systems, $L$ is typically very small \cite{heath2016overview}.}
\added{Static paths are defined by the absence of Doppler shifts, whereas dynamic paths exhibit nonzero Doppler shifts.}
We can express the CSI vector as follows, 
\begin{equation}
    \mathbf{h}_k(t) 
    = \underbrace{e^{j 2 \pi f_o(t)t}
    e^{-j 2 \pi \tau_o(t) f_k}}_{\eta_k(t)}
    \sum_{\ell = 0}^{L-1} 
    e^{-j 2 \pi \tau_{\ell} f_k}
    e^{j 2 \pi f_{D, \ell} t} 
    \beta_\ell 
    \mathbf{a}(\theta_{\ell}),
    \label{equ:CSI_k_subcarrier}
\end{equation}
where $\eta_k(t)$ represents the combined term of CFO and TO, $f_k = f_0+k\Delta f$ is the frequency of the $k$-th subcarrier, with $f_0$ being the initial frequency. 
The AoA of the $\ell$-th path is $\theta_{\ell}$, and the set of all AoAs is denoted as $\Theta$. 
The complex gain, propagation delay, and Doppler frequency of the $\ell$-th path are $\beta_{\ell}$,  $\tau_\ell = d_\ell/c$, and $f_{D, \ell} = v_\ell f_0/c$, respectively, where $d_\ell$ is the distance traveled by the signal, $v_\ell$ is radial speed of the reflector,  and $c$ is the speed of light.  
For the simplicity of expression, we assume that paths indexed from $\ell=0$ to $L_s-1$ are static, implying $f_{D, \ell} = 0, \forall \ell = 0, \cdots, L_s-1$, while the remaining paths are dynamic.
\added{Moreover, we assume that each path is characterized by a unique AoA, Doppler, and delay.}
Notably, the path with index $\ell = 0$ is defined as the LoS path and the others are NLoS paths. 
The steering vector $\mathbf{a}(\theta)$ for a half-wavelength spaced ULA with $N_r$ antennas is given by
\begin{equation}
    \mathbf{a}(\theta)=\left[
    1, e^{-j\pi \sin(\theta)}, \cdots, e^{-j\pi(N_r-1)\sin(\theta)}
    \right]^T.
    \label{equ:steering_vector}
\end{equation}

Assuming a BF matrix for the analog subarrays of BS, denoted as $\mathbf{W} = \left[\mathbf{w}_{\mathcal{A}}, \mathbf{w}_{\mathcal{B}} \right] \in\mathbb{C}^{N_r\times 2}$, with $\|\mathbf{w}\|_2^2=1$, the received signal vector after analog subarrays for the $k$-th subcarrier, can be represented as a vector $\mathbf{y}_k\in\mathbb{C}^{2\times 1}$:
\begin{equation}
    \begin{aligned}
        \mathbf{y}_{k}(t) &= \mathbf{W}^H\mathbf{r}_k(t)\\
        &= \mathbf{W}^H\mathbf{h}_k(t)x_k(t)+\mathbf{W}^H\mathbf{n}_k(t).
    \end{aligned}
    \label{equ:received_signal_after_BF_k_subcarrier}
\end{equation}

\begin{figure*}[tb!]
    \centering
    \includegraphics[width=\linewidth]{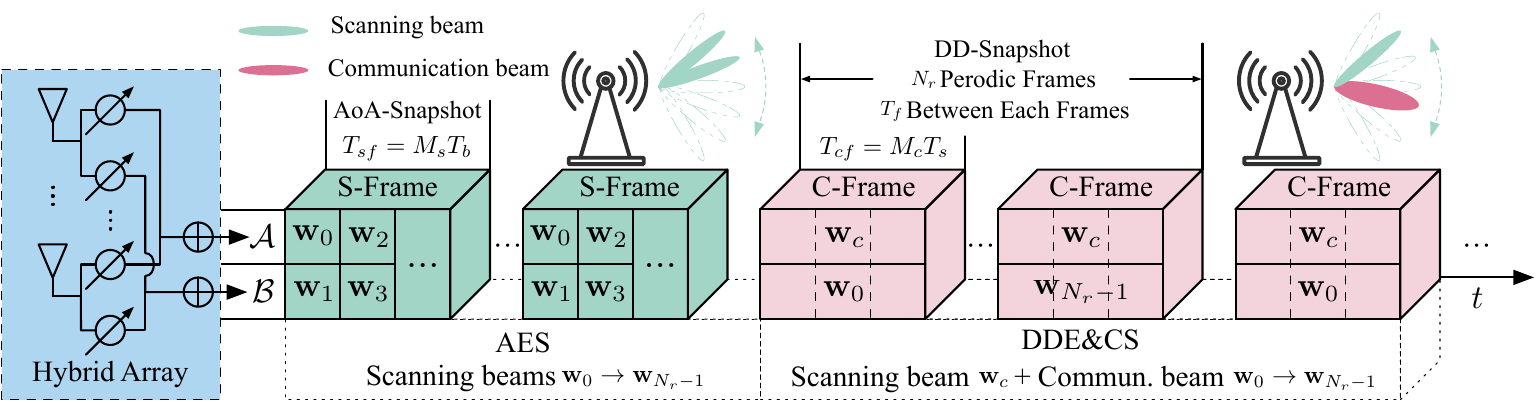}
    \caption{Illustration of the proposed ISAC framework, where $\mathbf{w}$ is a BF vector for the corresponding subarray at each timeslot. 
    S-frame and C-frame stand for frames used in AES and DDE\&CS, respectively.
    AoA-snapshot and DD-snapshot stand for snapshots used in  AoA estimation and Doppler-delay estimation, with the interval of these blocks representing the duration of collecting a snapshot.
    } 
    \label{fig:system_structure}
\end{figure*}

\subsection{\added{Basic Assumption of Phase Stability}}
\label{subsec:basic_assumption_of_phase_stability}
\added{
    In asynchronous ISAC systems, both clock variations and the parameters of dynamic objects are time-varying. 
    Nonetheless, reliable parameter estimation necessitates certain parameters to remain relatively stable over specific intervals. 
    To streamline our analysis, we adopt the following assumptions:
}
\begin{itemize}
    \item \added{\textbf{Assumption 1}: Doppler and delay are assumed to be constant within each millisecond-scale coherence processing interval (CPI)} \cite{zhang2021overview};
    \item \added{\textbf{Assumption 2}: CFO and TO are considered constant over a microsecond timescale.}
\end{itemize}

\added{
    As discussed in \cite{zhang2021overview}, a CPI is defined as the duration over which the parameters of the objects remain nearly constant. 
    The length of a CPI depends on the mobility of objects in the channel and typically spans a few milliseconds.
    In contrast, CFO and TO are characterized by more rapid variations due to the inherent instability of the local oscillator (LO).
    These parameters can be approximated as constant only over significantly shorter durations, typically on the order of microseconds.
    For instance, assuming an LO exhibits a typical clock stability of 10 parts-per-million (ppm) at 100MHz, the timing offset variation over 1us is around 10ps, corresponding to $ 10\text{ps} \times 100\text{MHz} = 1/1000$ of a clock cycle.
    This minor variation is insignificant for coherent processing that relies on phase information, even when extended to tens of microseconds.
}

\subsection{Proposed Framework}
\label{subsec:Proposed_Framework}
Next, we introduce the proposed ISAC framework, which includes two stages, as shown in Fig. \ref{fig:system_structure}. 
The figure also depicts the relationships among BF vectors, frames, and snapshots in each stage. 
The two stages are briefly described below, and will be detailed later.
\begin{itemize}
    \item \textit{\textbf{AoA Estimation Stage (AES).}} In this stage, the UE sequentially transmits multiple sensing frames (S-frames), each composed of pilot sequences, for AoA estimation. 
    The BS receives these signals using a BF matrix $\mathbf{W}$, where each BF vector ($\mathbf{w}_{\mathcal{A}}, \mathbf{w}_{\mathcal{B}}$) changes at a fixed interval $T_b$, typically spanning 1 to 4 OFDM symbols \cite{giordani2018tutorial} for beam scanning, which we refer to as a beam switching timeslot (BS-TS).
    Each S-frame comprises $M_s$ BS-TSs, forming a snapshot for AoA estimation.
    The time-varying BF vectors, denoted as $\mathbf{w}_b$, are selected from a predefined codebook and generate corresponding scanning beams. 
    The number of S-frames depends on the signal-to-noise ratio (SNR) level, e.g., more frames for lower SNR.
    A high-resolution AoA estimation algorithm is then executed to estimate the AoAs of all paths, followed by a procedure to identify the AoA of the LoS path. 
    It is worth noting that accurate estimation of AoAs for both LoS and NLoS paths is essential for subsequent Doppler-delay estimation.
    \added{In particular, in the DDE\&CS stage, the AoAs of both static and dynamic NLoS paths are further leveraged to design the BF vectors that mitigate the phase interference introduced by these paths; by filtering out such interference, the LoS path can be retrieved as a ``clean'' reference for aiding CFO and TO eliminating.}
    The detailed procedure of the AES is described in Section \ref{sec:AoA_estimation}. 
    
    \item \textit{\textbf{Doppler-Delay Estimation Plus Communication Stage (DDE\&CS).}} 
    In this stage, sensing does not require knowledge of the transmitted symbols and can be based on either unknown data symbols or pilots, which offers flexibility that ensures compatibility with current communication standards, such as 5G NR, LTE, etc.
    We refer to the signals transmitted during this stage as C-frames.
    As depicted in Fig. \ref{fig:system_structure}, $M_c$ and $T_f$ denote the frame length and the interval between frames in DDE\&CS. 
    The BF matrix $\mathbf{W}$ includes two types of beams: $\mathbf{w}_{\mathcal{A}} = \mathbf{w}_c$, the communication beam that also serves as a reference beam, and $\mathbf{w}_{\mathcal{B}} = \mathbf{w}_b$, the scanning beam. 
    Each snapshot for Doppler-delay estimation, referred to as a DD-snapshot, consists of data collected over $N_r$ consecutive C-frames.
    The scanning beam changes across frames but remains constant within a frame and is selected from the same predefined codebook used in AES, enabling subarray $\mathcal{B}$ to scan the entire spatial domain within each DD-snapshot.
    The communication beam remains unchanged unless the AoAs are updated. 
    Subsequently, the CFO and TO are eliminated using the concept of CSI-R, and Doppler and relative delay are finally estimated. 
    Comprehensive details of the DDE\&CS are provided in Section \ref{sec:estimation_of_doppler_and_delay}.
\end{itemize}

Reviewing the required signal structures in the two stages above, we can see that they comply well with conventional ones in mmWave communications.
Our proposed BF operation also has little impact on communications.
Therefore, our scheme offers great flexibility and compatibility when being implemented in current communication systems including mobile networks.

\section{AES: AoA Estimation}
\label{sec:AoA_estimation}

\subsection{Frame Structure for AoA Estimation}
\label{subsec:initial_beam_scanning_strategy}
In hybrid arrays, the reduced dimension of the received signal vector hinders AoAs estimation, necessitating beam scanning to recover spatial information. 
Consequently, we design a beam scanning method with a codebook to guide the scanning directions.

Following the beamspace representation \cite{heath2016overview}, we design the codebook for BF vectors using uniformly spaced spatial angles $u_{b} = b\Delta u = b/N_r, b = 0, \cdots, N_r-1$. 
The codebook $\mathcal{W}$ comprises steering vectors corresponding to $\theta_b = \arcsin(u_b/2)$, i.e., $\mathcal{W} = \{\mathbf{w}_0, \cdots, \mathbf{w}_{N_r-1}\}$, where $\mathbf{w}_b = \mathbf{a}(\theta_b)$. 
The length of the S-frame $M_s$ should be at least $N_r/2$ BS-TSs (we choose $M_s = N_r/2$) to ensure full spatial coverage. 
The BF vectors are switched sequentially from the codebook $\mathcal{W}$ during a  BS-TS period. 
Specifically, the BF matrix can be represented as $\mathbf{W}_m = \left[\mathbf{w}_{2m}, \mathbf{w}_{2m+1}\right], m = 0, \cdots, M_s-1$. 
Each S-frame is repeated $P$ times to gather sufficient data, specifically, $P$ AoA-snapshots, for accurate AoA estimation, with $P$ depending on the SNR conditions. 

\subsection{Frequency Smoothed MUSIC Algorithm for AoA Estimation}
\label{subsec:frequency_smoothing_for_coherent_signals}
\added{
    For the $m$-th BS-TS within a frame after CP and pilot removal, the received signal can be written in discrete form as
    \begin{align}
        \tilde{\mathbf{y}}_k[m] 
        &= \eta_k[m]\mathbf{W}_m^H \mathbf{A} \tilde{\mathbf{h}}_k[m] + \mathbf{W}_m^H\mathbf{n}_k[m],
        \label{equ:received_signal_after_BF_compact_form_h_k}
    \end{align}
    where 
    \begin{align}
        \eta_k[m] &= e^{j 2 \pi f_o[m]mT_b} e^{-j 2 \pi \tau_o[m] f_k},\label{equ:eta_TO_CFO}\\
        \mathbf{A} &= \left[\mathbf{a}(\theta_{0}), \cdots, \mathbf{a}(\theta_{L-1})\right],
        \label{equ:steering_matrix}
    \end{align}
    and
    \begin{equation}
        \tilde{\mathbf{h}}_k[m] =  \left[\begin{array}{c}
            e^{-j 2 \pi \tau_{0} f_k} e^{j 2 \pi f_{D, 0} mT_b}  \beta_0\\
            \vdots\\
            e^{-j 2 \pi \tau_{{L-1}} f_k} e^{j 2 \pi f_{D, {L-1}} mT_b}  \beta_{L-1}
        \end{array}\right],
        \label{equ:CSI_discrete_k_subcarrier}
    \end{equation}
    respectively. 
    Specifically, $\tilde{\mathbf{h}}_k[m]$, derived from (\ref{equ:CSI_k_subcarrier}), retains only the effective path information, including Doppler, delay, and complex gain, while excluding the common terms CFO and TO.
    $f_o[m]$ and $\tau_o[m]$ are CFO and TO in discrete form.
}

\added{
    During beam scanning, which involves $M_s$ sequential beam switches within a continuous scan cycle, the total scanning time for a single scan is typically on the order of several microseconds.
    For instance, both the literature \cite{sadhu202224} and commercial products, such as phased array antenna module (PAAM), report BS-TS as low as 8ns.
    Even when accounting for beam configuration time, the BS-TS remains as short as approximately 200ns.
    For an antenna array comprising 32 elements and 2 RF chains, the total duration for beam switching required for a single scan (i.e., the duration of one S-Frame) is given by $16\times 200\text{ns} = 3.2\text{us}$.
    Within this interval, the phase variation induced by LO instability remains negligible, as discussed in Section \ref{subsec:basic_assumption_of_phase_stability}.
}

\added{Accordingly, we assume $\tilde{\mathbf{h}}_k[m] \approx \tilde{\mathbf{h}}_k[m']$ and $\eta_k[m] \approx \eta_k[m'], \forall m\neq m', m, m' = 0, \cdots, M_s-1$.}
By stacking $\mathbf{y}_k[m]$ from $m=0$ to $m=M_s-1$ within an S-Frame, the snapshot vector $\boldsymbol{\varsigma}_k\in \mathbb{C}^{N_r\times 1}$ is constructed as
\added{
    \begin{equation}
        \begin{aligned}
            \boldsymbol{\varsigma}_k 
            &= \left[\begin{array}{c}
            \eta_k[0] \mathbf{W}_{0}^H\mathbf{A}\tilde{\mathbf{h}}_{k}[0]\\
            \vdots \\
            \eta_k[M_s-1] \mathbf{W}_{M_s-1}^H\mathbf{A}\tilde{\mathbf{h}}_{k}[M_s-1]
            \end{array}
            \right] + \mathbf{n}_{\boldsymbol{\varsigma}_k} \\
            &\approx \eta_k
            \mathbf{W}_s^H
            \mathbf{A}\tilde{\mathbf{h}}_{k} + \mathbf{n}_{\boldsymbol{\varsigma}_k},
        \end{aligned}
        \label{equ:measurement_vector_AoA_k_subcarrier}
    \end{equation}
    where $\tilde{\mathbf{h}}_{k} = \tilde{\mathbf{h}}_k[0]$, $\eta_k = \eta_k[0]$, 
}
\begin{equation}
    \mathbf{n}_{\boldsymbol{\varsigma}_k} =  \left[\begin{array}{c}
        \mathbf{W}_0^H\mathbf{n}_k[0] \\
        \vdots\\
        \mathbf{W}_{M_s-1}^H\mathbf{n}_k[M_s-1]
        \end{array}
        \right],
\end{equation}
and 
\begin{equation}
    \mathbf{W}_s = \left[\mathbf{W}_0, \cdots, \mathbf{W}_{M_s-1}\right] = 
    \left[\mathbf{w}_0, \cdots, \mathbf{w}_{N_r-1}\right],
    \label{equ:BF_matrix_all_beams}
\end{equation}
which is referred to as the stacked BF matrix. 
This scanning-and-stacking process expands the signal dimension to match the number of antennas, effectively forming a \textit{virtual sensor array} with a dimension of $N_r$. 
Additionally, the BF vectors $\mathbf{w}_b$ defined by the codebook are orthonormal, resulting in the covariance matrix of the stacked noise vector $\mathbf{n}_{\boldsymbol{\varsigma}_k}$ remaining $\sigma^2\mathbf{I}_{N_r}$.

\added{
    The covariance matrix of (\ref{equ:measurement_vector_AoA_k_subcarrier}) is therefore given by
    \begin{equation}
        \begin{aligned}
            \mathbf{R}_{\boldsymbol{\varsigma}_k} &= \mathbb{E}\left\{\boldsymbol{\varsigma}_k \boldsymbol{\varsigma}_k^H\right\}\\
            &= \mathbf{W}_s^H\mathbf{A}
            \tilde{\mathbf{H}}_k
            \mathbf{A}^H\mathbf{W}_s + 
            \sigma^2\mathbf{I}_{N_r},
        \end{aligned}
        \label{equ:covariance_matrix_AoA_k_subcarrier}
    \end{equation}
    where
    \begin{equation}
        \tilde{\mathbf{H}}_k = \mathbb{E}\left[|\eta_k|^2\tilde{\mathbf{h}}_{k}\left(\tilde{\mathbf{h}}_{k}\right)^H\right]
        \label{equ:H_k_weighted_average}
    \end{equation}
    represents a weighted expectation of the outer products of $\tilde{\mathbf{h}}_{k}$, with $|\eta_k|^2$ as the weighting factor.
}

\added{
    In our system design, all multipath signals originate from the same UE and share the same transmitted signal, indicating that the received signals from the $L_s$ static paths are coherent.
    Consequently, the rank of $\mathbf{R}_{\boldsymbol{\varsigma}_k}$ is $L_d + 1$ (see Appendix A), allowing the MUSIC algorithm to resolve up to $L_d + 1$ signals when using $\mathbf{R}_{\boldsymbol{\varsigma}_k}$.
    These resolved signals include $L_d$ dynamic paths and at most one combined static path (if a dominant LoS path exists).
    However, in this case, the estimated AoA of the combined static path may deviate from the actual LoS path direction due to the influence of other static paths.
}

To detect all static paths with high accuracy, we introduce a smoothing technique to restore the rank of $\tilde{\mathbf{H}}_k$ to $L$ ($L < N_r$). 
The smoothing is applied to the covariance matrix in the frequency domain and we refer to it as \textit{frequency smoothing}. 
This method maintains the full dimension of the virtual sensor array, effectively preventing the dimensional reduction that occurs with \textit{spatial smoothing} technique introduced by \cite{shan1985spatial}. 
The frequency smoothed covariance matrix can be written as
\added{
    \begin{equation}
        \begin{aligned}
            \overline{\mathbf{R}}_{\boldsymbol{\varsigma}} &= \frac{1}{K} \sum_{k=0}^{K-1} \mathbf{R}_{\boldsymbol{\varsigma}_k}\\
            &= \mathbf{W}_s^H\mathbf{A}\left(\frac{1}{K}\sum_{k=0}^{K-1}
            \tilde{\mathbf{H}}_k\right)
            \mathbf{A}^H\mathbf{W}_s + \sigma^2\mathbf{I}_{N_r}\\
            &= \mathbf{W}_s^H\mathbf{A}\mathbf{\overline{H}} \mathbf{A}^H\mathbf{W}_s + \sigma^2\mathbf{I}_{N_r},
        \end{aligned}
        \label{equ:covariance_matrix_AoA_smoothed}
    \end{equation}
    where $\mathbf{\overline{H}} = \frac{1}{K}\sum_{k=0}^{K-1}
            \tilde{\mathbf{H}}_k$. 
}

\added{
    \begin{proposition}
        If $\tau_i \neq \tau_j, \forall i \neq j$ and $K\geq L$, the rank of $\overline{\mathbf{H}}$ is $L$, regardless of the signal coherence. 
        \label{lem:rank_H_after_smoothing}
    \end{proposition}
    \begin{proof}
        See Appendix \ref{sec:appendix_B}.
    \end{proof}
}

\added{
    Base on Proposition \ref{lem:rank_H_after_smoothing}, $\overline{\mathbf{H}}$ retains rank $L$, which ensures the presence of an $L$-dimensional signal subspace within $\overline{\mathbf{R}}_{\boldsymbol{\varsigma}}$.
    Utilizing this property, we perform eigenvalue decomposition (EVD) on $\overline{\mathbf{R}}_{\boldsymbol{\varsigma}}$ to extract the AoAs of all multipaths.
}
The resulting eigenvector matrix $\mathbf{U}_\text{A} = \left[\mathbf{U}_{\text{A},s}\, \mathbf{U}_{\text{A},n}\right] \in \mathbb{C}^{N_r\times N_r}$ separates the signal subspace $\mathbf{U}_{\text{A},s}\in \mathbb{C}^{N_r\times L}$ and the noise subspace $\mathbf{U}_{\text{A},n}\in \mathbb{C}^{N_r\times (N_r-L)}$, with $\mathbf{U}_{\text{A},s}\perp \mathbf{U}_{\text{A},n}$.
The MUSIC spectrum is defined as
\begin{equation}
    P_{\text{AoA}}(\theta) = \frac{1}{
        \left\|\tilde{\mathbf{p}}^H(\theta)\mathbf{U}_{\text{A},n}\right\|_2^2
        },
    \label{equ:MUSIC_spectrum_AoA}
\end{equation}
where $\tilde{\mathbf{p}}(\theta) = \mathbf{W}_s^H\mathbf{a}(\theta)$ is a candidate vector. 
The MUSIC spectrum $P_{\text{AoA}}(\theta)$ exhibits $L$ peaks, each corresponding to an AoA of a path.
These estimated angles, denoted as $\tilde{\theta}_{\ell}$, form the set $\tilde{\Theta}$, encompassing both static and dynamic paths.

\subsection{LoS Path Determination}
\label{subsec:LoS_path_detection}
Identifying the accurate AoA of the LoS path is crucial for implementing efficient BF for communication and sensing.
This AoA is included in the set $\tilde{\Theta}$, from which the correct angle needs to be selected. 
One common approach selects the AoA by identifying the strongest beam direction collected during the beam scanning period and matching it with the closest angle in $\tilde{\Theta}$.
However, this method has low resolution due to the large scanning step size. 
We propose a LoS path determination method with higher resolution based on array response for hybrid arrays.

The power of the ideal array response for a ULA at an AoA of $\theta$ can be calculated by
\begin{equation}
    \begin{aligned}
        \rho(\theta_r, \theta) &= |\mathbf{w}(\theta_r)^H\mathbf{a}(\theta)|^2\\
        &=\frac{\sin^2\left(\frac{N(\pi \sin(\theta_r))-\pi\sin(\theta)}{2}\right)}{\sin^2\left(\frac{\pi \sin(\theta_r)-\pi\sin(\theta)}{2}\right)},
    \end{aligned}
    \label{equ:array_response_ideal}
\end{equation}
where the BF vector $\mathbf{w}(\theta)$ is defined as $\mathbf{a}(\theta)$, $\theta$ is the actual AoA of the received signal, and $\theta_r$ is the reference angle towards which the beamformer is steered. 
This function exhibits a sinc-like pattern and reaches its maximum power when $\theta_r = \theta$. 
As subcarriers do not affect LoS path detection, we assume the received signal from (\ref{equ:received_signal_after_antennas}) on the $k$-th subcarrier after CP and pilot removal, denoted as $\tilde{\mathbf{r}}_k[m]$, is available; the array response can then be represented as
\begin{equation}
    \begin{aligned}
        \left|\mathbf{w}^H (\theta_r) \tilde{\mathbf{r}}_k[m] \right|^2 &= \left|\mathbf{w}^H(\theta_r)
        \sum_{\ell = 0}^{L-1} 
        e^{j 2 \pi \left(f_{D, \ell} mT_s - \tau_{\ell} f_k\right)}
        \beta_\ell 
        \mathbf{a}(\theta_{\ell})\right|^2.
    \end{aligned}
    \label{equ:array_response_of_received_signal}
\end{equation}

With a dominant LoS path, where $\beta_0\gg \beta_\ell, \forall \ell = 1, \cdots, L-1$, the AoA for the LoS path can be estimated by finding the maximum value as follows,
\begin{equation}
    \begin{aligned}
        \tilde{\theta}_s &= \arg\max_{\theta_r} \left|\mathbf{w}(\theta_r)^H\tilde{\mathbf{r}}_k[m]\right|^2\\
        &= \arg\max_{\theta_r} \left|b_0\mathbf{w}(\theta_r)^H\mathbf{a}(\theta_{0})\right|^2.
    \end{aligned}
    \label{equ:LoS_detection_theta}
\end{equation}

However, given the implementation of hybrid arrays, we cannot directly obtain $\tilde{\mathbf{r}}_k[m]$. 
Hence, we extract $\tilde{\mathbf{r}}_k$ from the stacked received signal (\ref{equ:measurement_vector_AoA_k_subcarrier}). 
According to the definition of codebook $\mathcal{W}$, the stacked BF matrix $\mathbf{W}_s$ is a DFT matrix, meaning $\mathbf{W}_s\mathbf{W}_s^H = \mathbf{I}$. 
\added{Within an S-Frame, where the CFO and TO remain constant, we can extract $\tilde{\mathbf{r}}_k[m]$ by calculating}
\begin{equation}
    \tilde{\mathbf{r}}_k[m] \approx \mathbf{W}_s\boldsymbol{\varsigma}_k, \forall m = 0, \cdots, M_s-1. 
    \label{equ:received_signal_after_DFT}
\end{equation}
To simplify notation, we drop the index $m$ of (\ref{equ:received_signal_after_DFT}) in the subsequent analysis. 
The array response is then determined by performing a grid search over $\theta_r$ using (\ref{equ:LoS_detection_theta}). 
Denote $N'$ as the number of grids, and $\mathbf{w}(\theta)$ corresponds to a column in the DFT matrix when $\theta = \arcsin (u/2)$, where $u = n/N'$ and $n$ ranges from $0$ to $N'-1$. 
Consequently, the array response can be evaluated using an $N'$-point FFT on $\tilde{\mathbf{r}}_k$, with additional zero-padding applied to analyze a denser grid of potential $\theta$ values, yielding $\tilde{\theta}_s$. 
The estimated $\tilde{\theta}_s$ is then matched with the closest values in $\tilde{\Theta}$, from which the high-resolution AoA $\tilde{\theta}_0$ for the LoS path is finally selected. 
\added{The AoA estimation and LoS detection algorithm is summarized at Algorithm \ref{alg:AoA_LoS}.}

\begin{algorithm}[tb!]
    \caption{Proposed Algorithm for AoA Estimation and LoS Detection}
    \begin{algorithmic}[1]
        \STATE \textbf{Input:} $M_s$, $W_s$, $L$, $K$, number of grid points for AoA search $G_{\theta}$;

        \STATE \textbf{Initialization:} Set the estimated AoA vector $\tilde{\boldsymbol{\theta}}$ as zero vector of length $L$, and create a dictionary for the candidate vector $\tilde{\mathbf{p}}(\theta)$ by obtaining a uniform grid of the angular domain with $G_{\theta}$ points;

        \STATE Obtain $\boldsymbol{\varsigma}_k$ for $k = 0, \cdots, K-1$ and $P$ snapshots according to (\ref{equ:measurement_vector_AoA_k_subcarrier});

        \STATE Calculate $\overline{\mathbf{R}}_{\boldsymbol{\varsigma}}$ using (\ref{equ:covariance_matrix_AoA_k_subcarrier}) and (\ref{equ:covariance_matrix_AoA_smoothed});

        \STATE Compute the eigenvector matrix $\mathbf{U}_{\text{A}, s}$ corresponding to the $L$ largest eigenvalues of $\overline{\mathbf{R}}_{\boldsymbol{\varsigma}}$ using iterative methods;

        \STATE Obtain the MUSIC spectrum $\mathbf{P}_{\text{AoA}}$ according to (\ref{equ:MUSIC_spectrum_AoA}) by evaluating candidate vectors $\tilde{\mathbf{p}}(\theta)$;

        \STATE Search for the $L$ largest peaks of $\mathbf{P}_{\text{AoA}}$, obtain the corresponding angles $\tilde{\boldsymbol{\theta}}$; \hfill \makebox[0cm][r]{\textcolor{gray}{\(\triangleright\) \textit{AoA Estimation}}}

        \STATE Calculate $\tilde{\mathbf{r}}_k$ using $\boldsymbol{\varsigma}_k$ according (\ref{equ:received_signal_after_DFT});

        \STATE Evaluate (\ref{equ:LoS_detection_theta}) by grid search over $\theta$ using $N^\prime$-points FFT and identify the angle $\tilde{\theta}_s$ of the maximum value;

        \STATE Match $\tilde{\theta}_s$ with the closest value from $\tilde{\boldsymbol{\theta}}$, define it as the LoS angle $\theta_0$, and form the AoA set $\tilde{\Theta}$ with the remaining values of $\tilde{\boldsymbol{\theta}}$; \hfill \makebox[0cm][r]{\textcolor{gray}{\(\triangleright\) \textit{LoS Detection}}}

        \STATE \textbf{Output:} AoAs set $\tilde{\Theta}$.
    \end{algorithmic}
    \label{alg:AoA_LoS}
\end{algorithm}

\section{DDE\&CS: Delay and Doppler Estimation}
\label{sec:estimation_of_doppler_and_delay}
In a bi-static configuration, conducting both communication and sensing simultaneously in a communication system equipped with a hybrid array presents significant challenges. 
These challenges primarily arise from three key issues. 

\begin{itemize}
    \item \textbf{Presence of CFO and TO.} 
    As previously discussed, clock asynchrony in bi-static configurations inevitably causes CFO and TO, leading to ambiguity in Doppler and delay estimation. 
    Unlike in communication processes, where CFO and TO are jointly canceled with Doppler and delay, respectively, sensing requires their independent cancellation to ensure accurate Doppler and delay estimation for each path \cite{wu2024sensing}. 
    Therefore, a real-time mitigation of CFO and TO is essential. 

    \item \textbf{Sensing range limitations due to directional beam gain.}  
    In our Doppler and delay estimation period, data is obtained after analog subarrays with directional beams that amplify signals from specific directions. 
    If targets fall outside the coverage of these beams, the gain drops significantly, resulting in low SNR and inaccurate estimations for Doppler and delay.
    While the communication beam provides sufficient beam gain for communication, it lacks comprehensive sensing coverage. 
    Therefore, scanning beams covering multiple directions are essential for expanding the sensing range and maintaining a high SNR for accurate Doppler and delay estimation during this stage.

    \item \textbf{Sensing through data frames.} 
    Before communication, channel estimation typically uses pilot sequences, such as the demodulation reference signal (DMRS), to obtain CSI for signal equalization.
    While useful for sensing, these sequences are insufficient in volume to gather the necessary data for robust and high-resolution Doppler and delay estimation.
    Moreover, allocating significant resources to transmit these sequences solely for sensing reduces communication efficiency.
    Therefore, using data frames for sensing without disrupting the communication process is a more efficient approach, albeit challenging.
\end{itemize}

To address the aforementioned challenges, we propose a practical scheme, as detailed in the following subsections.

\subsection{Proposed BF Vector Designs for Communication Beam}
\label{subsec:proposed_combining_vector_for_communication}
In our hybrid array configuration, two subarrays operate independently during DDE\&CS. 
As described in Section \ref{sec:system_model}, subarray $\mathcal{A}$ generates a communication beam, denoted as $\mathbf{w}_c$, and subarray $\mathcal{B}$ generates a scanning beam.  
The BF vector for the scanning beam is also obtained by referring to the codebook sequentially, as mentioned in Section \ref{subsec:initial_beam_scanning_strategy}, i.e., $\mathbf{w}_{\mathcal{B}} \in \mathcal{W}$. 

The BF vector of the communication beam $\mathbf{w}_c$ is designed to achieve two primary objectives:
\begin{itemize}
    \item \textbf{High beam gain for communication.} 
    The beam should achieve high gain towards the LoS path to ensure robust communication performance. 
    This involves optimizing the BF vector for different objectives, e.g., maximizing the signal-to-interference-plus-noise ratio (SINR).
    \item \textbf{Reference signal for CFO and TO elimination.} 
    The beam also serves as a reference for eliminating CFO and TO during DDE\&CS. 
    The objective is to generate a signal that contains only CFO, TO, and a predictable static phase offset.
    This is achieved by designing a BF vector that effectively suppresses all other paths, thereby retaining only a single static path. 
\end{itemize}

For notational simplicity, we focus on the $k$-th subcarrier and omit the combined CFO and TO term $\eta$ in this subsection, as these factors do not impact the results of our proposed BF designs.
\added{In this BF design, we adopt energy constraints to simplify the problem formulation, while the optimization algorithm design is beyond the scope of this work.}
\added{This constraint remains practically relevance given modern hybrid architectures (e.g., active arrays with phase shifters and variable gain amplifiers (VGAs)) that enable amplitude control in analog/hybrid BF.}
\added{Various advanced BF techniques (such as Bartlett, minimum variance distortionless response (MVDR), and other BF approaches) can be leveraged to generate reference signals for mitigating CFO and TO.}
\added{Here, we propose several exemplified BF schemes with diversity tradeoffs among performance, complexity, and applications, including a conventional Bartlett beamformer that maximizes the desired signal power without NLoS compression, a null-space approach that suppresses NLoS paths, and two intermediate approaches that balance gain and interference suppression.}

\subsubsection{Bartlett beamformer}
Based on the identified AoA of the LoS path during the AES, the steering vector for this AoA is directly used to form the BF vector: 
\begin{equation}
    \hat{\mathbf{w}}_{\text{BB}} = \mathbf{a}(\tilde{\theta}_0).
\end{equation}
This method, known as the \textit{Bartlett beamformer}, is noted for its simplicity and effectiveness, making it suitable for communication. 

\subsubsection{Null-Space (NS) approach}
Bartlett beamformer's limited ability to suppress NLoS paths can lead to interference, which sometimes results in estimation failures when used as a reference. 
To ensure a reference free from NLoS path interference, we adopt a method based on the principle of NS projection, leveraging spatial filtering to eliminate unwanted signals. 
The estimated steering matrix for the NLoS paths $\tilde{\mathbf{A}}_{\text{N}} \in \mathbb{C}^{N_r\times (L-1)}$, is constructed as follows,
\begin{equation}
    \tilde{\mathbf{A}}_{\text{N}} = \left[\mathbf{a}(\tilde{\theta}_{1}), \cdots, \mathbf{a}(\tilde{\theta}_{L-1})\right].
\end{equation}
The desired beam is formed by maximizing the received signal power within the NS of the NLoS paths. 
The optimization problem can be formulated as 
\begin{equation}
    \begin{aligned}
        \max_{\mathbf{w}} \quad & \left|\mathbf{w}^H\tilde{\mathbf{r}}_k\right|^2 \\
        \text { s.t. } \quad & \mathbf{w}^H \tilde{\mathbf{A}}_{\text{N}} = \mathbf{0}, \quad \left\|\mathbf{w}\right\|_2^2 = 1,
    \end{aligned}
    \label{equ:opt_NS_1}
\end{equation}
Let $\boldsymbol{\Psi}$ be the NS matrix of $\tilde{\mathbf{A}}_{\text{N}}^H$, i.e., $\boldsymbol{\Psi}^H\tilde{\mathbf{A}}_{\text{N}} = \mathbf{0}$. 
Letting $\boldsymbol{\alpha}$ be the weight vector, the original optimization problem can be reformulated as follows,
\begin{equation}
    \begin{aligned}
        \max_{\boldsymbol{\alpha}} \quad & \left|\boldsymbol{\alpha}^H\boldsymbol{\Psi}^H\tilde{\mathbf{r}}_k \right|^2 \\
        \text { s.t. } \quad & \left\|\boldsymbol{\Psi}\boldsymbol{\alpha}\right\|_2^2 = 1.
    \end{aligned}
\end{equation}

This is a Rayleigh quotient problem, which is solved using the method of Lagrange multipliers. 
The optimal $\boldsymbol{\alpha}$ is the eigenvector corresponding to the largest eigenvalue of $\boldsymbol{\Psi}^H \tilde{\mathbf{r}}_k \tilde{\mathbf{r}}_k^H \boldsymbol{\Psi}$. 
The BF vector is then derived as $\hat{\mathbf{w}}_{\text{NS}} = \mathbf{\Psi}\boldsymbol{\alpha}$.

\subsubsection{Hybrid approach}
The NS approach effectively suppresses NLoS paths but sometimes results in low gain for the LoS path, especially when multiple paths have close AoAs, leading to communication challenges. 
Inspired by multibeam design \cite{zhang2018multibeam}, we propose a hybrid strategy combining these two methods. 
The BF vector is constructed using an energy distribution factor $\rho$, $0\leq \rho \leq 1$ and a phase shift $\varphi$ as follows,
\begin{equation}
    \hat{\mathbf{w}}_{\text{HB}} = \sqrt{\rho}\hat{\mathbf{w}}_{\text{BB}} + e^{-j\varphi}\sqrt{1-\rho}\hat{\mathbf{w}}_{\text{NS}},
\end{equation}
where $\hat{\mathbf{w}}_{\text{BB}}$ and $\hat{\mathbf{w}}_{\text{NS}}$ are derived from the Bartlett beamformer and NS approaches, respectively. 
The phase shift $\varphi$ ensures coherent combination towards the LoS path by setting $\varphi$ such that $\angle \hat{\mathbf{w}}^H_{\text{BB}}\mathbf{a}(\tilde{\theta}_0) = \angle e^{j\varphi}\hat{\mathbf{w}}^H_{\text{NS}}\mathbf{a}(\tilde{\theta}_0) $.
The power of $\hat{\mathbf{w}}_{\text{HB}}$ is further normalized to 1. 
The energy distribution factor $\rho$ balances communication and sensing performance, which will be detailed in Section \ref{subsec:performance_beam}.

\subsubsection{SINR optimization approach}
Although the hybrid approach is straightforward, it still faces challenges.
A high $\rho$ can result in NLoS leakage, whereas a low $\rho$ may lead to insufficient LoS gain, making it difficult to select an optimal $\rho$ for balancing both sensing and communication in practice.

To address this, we propose a BF vector design that optimizes the SINR based on the estimated AoAs, achieving a robust balance between communication and sensing performance while eliminating the need for selecting $\rho$ in specific scenarios. 

Taking the NLoS paths as interference and ignoring the terms that do not impact power, the SINR of the received signal using $\mathbf{w}$ is
\begin{equation}
    \begin{aligned}
        \text{SINR} &= \frac{
          |\mathbf{w}^H\beta_0\mathbf{a}(\theta_0)|^2
        }{
          |\mathbf{w}^H\mathbf{A}_\text{N}\boldsymbol{\beta}_\text{N}+\mathbf{w}^H\mathbf{n}|^2
        }\\
        &= \frac{
          \beta_0^2\mathbf{w}^H\mathbf{a}(\theta_0)\mathbf{a}^H(\theta_0)\mathbf{w}
        }{
        \mathbf{w}^H \left(\mathbf{A}_\text{N} \boldsymbol{\beta}_\text{N} \boldsymbol{\beta}_\text{N}^H \mathbf{A}_\text{N}^H+ \mathbf{R}_\mathbf{n} \right)\mathbf{w}
        },
    \end{aligned}
\end{equation}
where $\boldsymbol{\beta}_\text{N} = [\beta_1, \cdots, \beta_{L-1}]^T$ represents the complex gains of the NLoS paths. 
During the AES, $\mathbf{a}(\theta_0)$ and $\mathbf{A}_{\text{N}}$ are estimated as $\tilde{\mathbf{A}} = [\mathbf{a}(\tilde{\theta}_0), \tilde{\mathbf{A}}_{\text{N}}]$. 
The complex gain vector $\tilde{\boldsymbol{\beta}} = [\tilde{\beta}_0; \tilde{\boldsymbol{\beta}}_{\text{N}}]$ can be obtained by $\tilde{\mathbf{r}}_k$ using the Least Squares (LS) estimator:
\begin{equation}
    \min_{\boldsymbol{\beta}}\|\tilde{\mathbf{r}}_k-\tilde{\mathbf{A}} \boldsymbol{\beta}\|_2^2,
\end{equation}
where $\tilde{\mathbf{r}}_k$ is derived from (\ref{equ:received_signal_after_DFT}). 
The optimal solution is $\tilde{\boldsymbol{\beta}} = \tilde{\mathbf{A}}^\dagger \tilde{\mathbf{r}}_k$. 
Additionally, the noise power $\sigma^2$ can be estimated by averaging the $N_r-L$ smallest eigenvalues of $\tilde{\mathbf{r}}_k \tilde{\mathbf{r}}_k^H$, thereby forming the noise covariance matrix $\mathbf{R}_{\mathbf{n}} = \sigma^2\mathbf{I}$. 

Subject to the power constraint, the SINR optimization problem is
\begin{equation}
    \begin{aligned}
        \max_{\mathbf{w}} \quad& \frac{
        \mathbf{w}^H\mathbf{R}_{0}\mathbf{w}
        }{
        \mathbf{w}^H\mathbf{R}_\text{N}\mathbf{w}
        }\\
        \text{s.t.} \quad & \|\mathbf{w}\|_2^2=1,
    \end{aligned}
\end{equation}
where $\mathbf{R}_{0} = b_0^2\mathbf{a}(\theta_0)\mathbf{a}^H(\theta_0)$ and $\mathbf{R}_\text{N} = \mathbf{A}_\text{N} \boldsymbol{\beta}_\text{N} \boldsymbol{\beta}_\text{N}^H \mathbf{A}_\text{N}^H+ \mathbf{R}_\mathbf{n}$. 
This is a generalized Rayleigh quotient problem, with the closed-form solution $\mathbf{w}_{\text{SINR}}$ being the left singular vector corresponding to the largest singular value of $\mathbf{R}_\text{N}^{-1}\mathbf{R}_0$. 

Based on the proposed beam designs, the communication beam's BF vector, denoted as $\mathbf{w}_c$, can be selected from $\mathbf{w}_{\text{BB}}$, $\mathbf{w}_{\text{NS}}$, $\mathbf{w}_{\text{HB}}$ or $\mathbf{w}_{\text{SINR}}$ according to specific requirements.

\subsection{Elimination of CFO and TO}
\label{subsec:elimination_of_TO_and_CFO}
Let $m$, $b$, and $g$ represent the indices of OFDM symbol within each C-frame, the frame within each snapshot, and the snapshot, respectively. 
Using the BF matrix $\mathbf{W}_b = \left[\mathbf{w}_{c}, \mathbf{w}_b\right]$ and neglecting noise, 
the received signal vector of the $b$-th frame at the $m$-th OFDM symbol is represented as
\begin{equation}
    \begin{aligned}
        \mathbf{y}_k[b, m] &= \eta_k[b,m]\mathbf{W}_b^H \mathbf{A}\tilde{\mathbf{h}}_k[b,m]x_k[b,m]\\
        &= \eta_k[b,m]\left[\begin{array}{c}
            y_{c, k}[b, m]\\
            y_{s, k}[b, m]
        \end{array}\right],
    \end{aligned}
\end{equation}
where $y_{c, k}[b, m]$ and $y_{s, k}[b, m]$ stands for the output of subarray $\mathcal{A}$ and $\mathcal{B}$, respectively. 

Referring to the discussion in Section \ref{subsec:proposed_combining_vector_for_communication}, \added{both static and dynamic NLoS paths can be neglected under the communication beam designs, we have}
\begin{equation}
    \begin{aligned}
        y_{c, k}[b, m] &= \eta_k[b,m]\mathbf{w}_c^H\mathbf{A}\tilde{\mathbf{h}}_k[b,m]x_k[b,m]\\
            &\approx \eta_k[b,m] \beta_0 e^{-j 2 \pi \tau_{0} (f_0+k\Delta f)}\mathbf{w}_c^H \mathbf{a}(\theta_0) x_k[b,m]\\
            &= \eta_k[b,m] \alpha_0^{-1} e^{-j 2 \pi \tau_{0} k\Delta f} x_k[b,m],
    \end{aligned}
    \label{equ:commun_signal_approx}
\end{equation}
where $\alpha_0^{-1} = \beta_0e^{-j 2 \pi \tau_{0} f_0}\mathbf{w}_c^H\mathbf{a}(\theta_0)$ is a constant value.

\added{
    Since the duration of each frame is much shorter than a CPI, we assume that Doppler remains constant within each frame, i.e., $e^{j2\pi f_{D, \ell}m_1 T_s} \approx e^{j2\pi f_{D, \ell}m_2 T_s}, \forall m_1\neq m_2, m_1, m_2 = 0, \cdots, M_c-1$.}
\added{Under this assumption, $y_{s, k}[b, m]$ can be expressed} in a compact form as
\begin{equation}
    \begin{aligned}
        y_{s, k}[b, m] 
        &= \eta_k[b,m]\mathbf{w}_b^H\mathbf{A}\boldsymbol{\Phi}_\tau^k \boldsymbol{\Phi}_D^{b} \boldsymbol{\beta}_{f_0}x_k[b, m],
    \end{aligned}
\end{equation}
where
\begin{align}
    \boldsymbol{\Phi}_\tau &= \operatorname{diag}\left(\left[e^{-j2\pi \tau_0 \Delta f}, \cdots, e^{-j2\pi \tau_{L-1} \Delta f}\right]\right),\\
    \boldsymbol{\Phi}_D &= \operatorname{diag}\left(\left[e^{j2\pi f_{D, 0} T_f}, \cdots, e^{j2\pi f_{D, L-1} T_f}\right]\right),\\
    \boldsymbol{\beta}_{f_0} &= \left[
    \beta_0e^{-j 2 \pi \tau_{0} f_0},
    \cdots,
    \beta_{L-1}e^{-j 2 \pi \tau_{L-1} f_0}
    \right]^T,
\end{align}
and $f_{D, \ell} = 0, \forall \ell = 0, \cdots, L_s-1$.

To mitigate CFO and TO, we calculate the signal ratio of $y_{s,k}[b,m]$ to $y_{c,k}[b,m]$ as follows,
\begin{equation}
    \begin{aligned}
        \zeta_k[b,m] &= \frac{y_{s,k}[b,m]}{y_{c,k}[b,m]}\\
        &= \alpha_0 \mathbf{w}_b^H\mathbf{A}\tilde{\boldsymbol{\Phi}}_\tau^k \boldsymbol{\Phi}_D^b \boldsymbol{\beta}_{f_0},
    \end{aligned}
    \label{equ:ratio_of_received_signal}
\end{equation}
where $\tilde{\boldsymbol{\Phi}}_\tau = \operatorname{diag}\left(\left[e^{-j2\pi (\tau_0-\tau_0) \Delta f}, \cdots, e^{-j2\pi (\tau_{L-1}-\tau_0) \Delta f}\right]\right)$ represents the relative delay matrix with respect to the LoS path. 
By observing $\zeta_k[b,m]$, we identify that it contains AoA, Doppler, and delay information, with the symbols removed. 
This characteristic allows for estimating these parameters independent of the transmitted symbols, thereby enabling the use of data payload frames for sensing while ensuring compatibility with existing communication protocols.

\subsection{Proposed AoA-Based Joint Doppler-Delay Estimation}
In this section, we present an algorithm for estimating Doppler and delay in a hybrid array system that leverages the estimated AoAs and uses a 2D-FFT to enhance the computational efficiency of MUSIC spectrum analysis. 

From (\ref{equ:ratio_of_received_signal}), we observe that Dopplers lead to phase variations across frames, while relative delays lead to phase variations across subcarriers. 
Based on these properties, we can stack a matrix $\boldsymbol{\Xi}_{\tilde{k}}[m]\in\mathbb{C}^{ \tilde{K}\times N_r}$ as follows,
\begin{equation}
    \boldsymbol{\Xi}_{\tilde{k}}[m] = \left[
        \begin{array}{ccc}
            \zeta_{\tilde{k}}[0, m] & \cdots & \zeta_{\tilde{k}}[N_r-1, m]\\
            \vdots & & \vdots\\
            \zeta_{\tilde{k} + \tilde{K}-1}[0, m] & \cdots & \zeta_{\tilde{k} + \tilde{K}-1}[N_r-1, m]
        \end{array}
    \right],
    \label{equ:stacked_measurement_matrix}
\end{equation}
where $\tilde{k} = 0, \cdots, \tilde{K}$ and $\tilde{K} = K/2$. 

Denote $\boldsymbol{\xi}_{\tilde{k}}[m]\in\mathbb{C}^{N_r\tilde{K} \times 1}$ as the vectorized $\boldsymbol{\Xi}_{\tilde{k}}[m]$, i.e., 
\begin{equation}
    \begin{aligned}
        \boldsymbol{\xi}_{\tilde{k}}[m] &= \operatorname{vec}\{\boldsymbol{\Xi}_{\tilde{k}}[m]\}\\
    \end{aligned},
    \label{equ:vectorized_measurement_matrix}
\end{equation}
which is a snapshot vector for estimation. 
The mean of $\boldsymbol{\xi}_{\tilde{k}}[m]$ within each frame is calculated as
\begin{equation}
    \overline{\boldsymbol{\xi}}_{\tilde{k}} = \frac{1}{M_c}\sum_{m=0}^{M_c-1}\boldsymbol{\xi}_{\tilde{k}}[m],
    \label{equ:mean_vectorized_measurement_matrix}
\end{equation}
which we term \textit{frame averaging}, can effectively reduce noise by averaging the CSI within a CPI, regardless of the communication symbols.

The covariance matrix of $\overline{\boldsymbol{\xi}}_{\tilde{k}}$ is given by
\begin{equation}
    \mathbf{R}_{\overline{\boldsymbol{\xi}}_{\tilde{k}}} = \mathbb{E}\left\{ \overline{\boldsymbol{\xi}}_{\tilde{k}} \overline{\boldsymbol{\xi}}_{\tilde{k}}^H \right\},
    \label{equ:covariance_matrix_doppler_delay_before_smoothing}
\end{equation}
where $\mathbf{R}_{\overline{\boldsymbol{\xi}}_{\tilde{k}}}\in\mathbb{C}^{N_r\tilde{K} \times N_r\tilde{K}}$.

To address the rank reduction of $\mathbf{R}_{\overline{\boldsymbol{\xi}}_{\tilde{k}}}$, we apply frequency smoothing, as outlined in Section \ref{subsec:frequency_smoothing_for_coherent_signals}. 
The frequency smoothed covariance matrix is
\begin{equation}
    \begin{aligned}
        \overline{\mathbf{R}}_{\overline{\boldsymbol{\xi}}} &= \frac{1}{\tilde{K}}\sum_{\tilde{k}=0}^{\tilde{K}-1}\mathbf{R}_{\overline{\boldsymbol{\xi}}_{\tilde{k}}}.
    \end{aligned}
    \label{equ:covariance_matrix_doppler_delay_smoothed}
\end{equation}
Following the approach in Section \ref{subsec:frequency_smoothing_for_coherent_signals}, we extract the signal subspace $\mathbf{U}_{\text{D},s}$ and the noise subspace $\mathbf{U}_{\text{D},n}$ from the eigenvector matrix $\mathbf{U}_\text{D} = \left[\mathbf{U}_{\text{D},s}\,\, \mathbf{U}_{\text{D},n}\right]$ of $\overline{\mathbf{R}}_{\overline{\boldsymbol{\xi}}}$, where $\mathbf{U}_{\text{D},s}$ corresponds to the $L$ largest eigenvalues.

The basis vectors $\mathbf{b}_{1}(\theta, f)\in \mathbb{C}^{N_r\times 1}$ for AoA and Doppler estimation are defined as
\begin{equation}
    \mathbf{b}_{1}(\theta, f) = 
    \left(\mathbf{W}_s^H \mathbf{a}(\theta)\right)
    \odot
    \mathbf{b}_0(f)
    ,
    \label{equ:basis_vector_Dopplers_AoAs}
\end{equation}
where $\mathbf{W}_s$ is defined as (\ref{equ:BF_matrix_all_beams}) and 
\begin{equation}
    \mathbf{b}_0(f) = \left[
        1,
        e^{j2\pi fT_f},
        \cdots,
        e^{j2\pi (N_r-1)fT_f}
    \right]^T,
\end{equation}
respectively, and for delay estimation, the basis vectors $\mathbf{b}_{2}(\tau)\in\mathbb{C}^{\tilde{K}\times 1}$ are defined as
\begin{equation}
    \mathbf{b}_{2}(\tau) = \left[1, e^{-j2\pi \tau \Delta f}, \cdots, e^{-j2\pi \tau (\tilde{K}-1)\Delta f}\right]^T.
    \label{equ:basis_vector_tau}
\end{equation}

From the form of matrix $\overline{\mathbf{R}}_{\boldsymbol{\xi}}[m]$, the joint basis vector $\mathbf{b}_{\text{D}}(\theta, f, \tau)\in\mathbb{C}^{N_r\tilde{K}\times 1}$ is given by
\begin{equation}
    \mathbf{b}_{\text{D}}(\theta, f, \tau) = \mathbf{b}_{1}(\theta, f) \otimes \mathbf{b}_{2}(\tau).
    \label{equ:basis_vector_combined}
\end{equation}

Using the noise subspace, the 3D-MUSIC spectrum is
\begin{equation}
    P_{\text{D}}(\theta, f, \tau) = \frac{1}{\|\mathbf{b}^H_{\text{D}}(\theta, f, \tau)\mathbf{U}_{\text{D}, n}\|_2^2}.
    \label{equ:MUSIC_spectrum_for_DD_estimation_noise_space}
\end{equation}
A grid search over $\theta$, $f$, and $\tau$ identifies $L$ peaks, corresponding to the parameters of the respective paths. 
However, joint estimation of these three parameters in the 3D-MUSIC spectrum results in high computational complexity. 
To address this, we use pre-estimated AoAs $\tilde{\Theta}$ from Section \ref{subsec:frequency_smoothing_for_coherent_signals} as priors, Doppler and delay can be estimated simultaneously by finding the peak of a 2D surface:
\begin{equation}
    \begin{aligned}
        \tilde{f}_\ell, \tilde{\tau}_\ell = \arg\max_{f, \tau} \frac{1}{\|\mathbf{b}^H_{\text{D}}(\tilde{\theta}_\ell, f, \tau)\mathbf{U}_{\text{D}, n}\|_2^2}, \forall \ell = 1, \cdots, L-1.
    \end{aligned}
\end{equation}
This reduces the search complexity to that of 2D-MUSIC. 

Nevertheless, the combined dimensions of Doppler and delay remain significant. 
With $G_D$ and $G_\tau$ representing the number of candidate estimates for Doppler and delay, respectively, the search dimension expands to $G_DG_\tau$.
Additionally, the matrix $\mathbf{U}_{\text{D},n}$ has dimensions $N_r\tilde{K}\times\left(N_r\tilde{K}-L\right)$, resulting in substantial computational load for estimation. 
We note that the left singular matrix $\mathbf{U}_\text{D}$ is a unitary matrix of size $N_r\tilde{K}\times N_r\tilde{K}$, satisfying $\|\mathbf{b}^H\mathbf{U}_\text{D}\|_2^2 = \|\mathbf{b}^H\|_2^2$. 
Thus, (\ref{equ:MUSIC_spectrum_for_DD_estimation_noise_space}) can be reformulated as
\begin{equation}
    \begin{aligned}
        P_{\text{D}}(\theta, f, \tau) &= \frac{1}{\|\mathbf{b}^H_{\text{D}}(\theta, f, \tau)\|_2^2 - \|\mathbf{b}^H_{\text{D}}(\theta, f, \tau)\mathbf{U}_{\text{D}, s}\|_2^2}.
    \end{aligned}
\end{equation}
By normalizing the basis vector $\mathbf{b}^H_{\text{D}}(\theta, f, \tau)$, we ensure that $\|\mathbf{b}^H_{\text{D}}(\theta, f, \tau)\|_2^2 = N_r\tilde{K}$ remains constant.

The dimension of $\mathbf{U}_{\text{D}, s}$ is $N_r\tilde{K}\times L$.
Since $L\ll N_r\tilde{K}$, it is significantly smaller than the dimension of $\mathbf{U}_{\text{D}, n}$. Consequently, searching the MUSIC spectrum in the signal subspace $\mathbf{U}_{\text{D},s}$ is much more efficient than in the noise subspace.
Furthermore, computing the eigenvector matrix of the symmetric matrix $\overline{\mathbf{R}}_{\overline{\boldsymbol{\xi}}}$ for the signal subspace is more computationally efficient, especially when using an iterative method such as the Lanczos method.

To further accelerate the evaluation of the MUSIC spectrum, we adopt a 2D-FFT approach. 
First, we partition $\mathbf{U}_{\text{D},s}$ into
\begin{equation}
    \mathbf{U}_{\text{D},s} = \left[\mathbf{u}_{0}, \cdots, \mathbf{u}_{L-1}\right].
    \label{equ:Space_Partition}
\end{equation}

Then, we have
\begin{equation}
    P_{\text{D}}(\theta, f, \tau) = \frac{1}{N_r\tilde{K}-\sum_{i=0}^{L-1}\|\mathbf{b}^H_{\text{D}}(\theta, f, \tau)\mathbf{u}_{i}\|_2^2}.
    \label{equ:MUSIC_Spectrum_Signal_Subspace}
\end{equation}

We observe that 
\begin{equation}
    \begin{aligned}
        &\quad\mathbf{b}^H_{\text{D}}(\theta, f, \tau)\mathbf{u}_{i} \\
        &= \left\{\left[\mathbf{W}_s^H\mathbf{a}(\theta)\odot \mathbf{b}_0(f)\right]\otimes \mathbf{b}_2(\tau)\right\}^H\mathbf{u}_{i}\\
        &= \left\{\left[\mathbf{W}_s^H\mathbf{a}(\theta)\otimes \mathbf{1}_{\tilde{K}\times 1}\right]\odot\left[\mathbf{b}_0(f)\otimes \mathbf{b}_2(\tau)\right]\right\}^H\mathbf{u}_{i}\\
        &= \left[\mathbf{b}^H_0(f)\otimes \mathbf{b}^H_2(\tau)\right]\left\{\left[\mathbf{W}_s^H\mathbf{a}(\theta)\otimes \mathbf{1}_{\tilde{K}\times 1}\right]\odot \mathbf{u}_i\right\}\\
        &\overset{(a)}{=} \mathbf{b}^H_2(\tau)\tilde{\mathbf{U}}_{i, \theta}\mathbf{b}^*_0(f),
    \end{aligned}
    \label{equ:2D_FFT_derivation}
\end{equation}
where condition $(a)$ is satisfied by noting
\begin{equation}
    \operatorname{vec}\{\mathbf{ABC}\}=(\mathbf{C}^T\otimes \mathbf{A})\operatorname{vec}\{\mathbf{B}\}. 
\end{equation}

Furthermore, $\tilde{\mathbf{U}}_{i, \theta}$ is the matrix obtained by reshaping the vector $\left[\mathbf{W}_s^H\mathbf{a}(\theta)\otimes \mathbf{1}_{\tilde{K}\times 1}\right]\odot \mathbf{u}_i$, that is, 
\begin{equation}
    \operatorname{vec}\left\{\tilde{\mathbf{U}}_{i, \theta}\right\}=\left[\mathbf{W}_s^H\mathbf{a}(\theta)\otimes \mathbf{1}_{\tilde{K}\times 1}\right]\odot \mathbf{u}_i.
\end{equation}

In equation (\ref{equ:2D_FFT_derivation}), both $\mathbf{b}_2(\tau)$ and $\mathbf{b}_0(f)$ are columns in a DFT matrix. 
Therefore, we define candidate parameters $\tau = 0, \frac{1}{\tilde{K}\Delta f}, \cdots, \frac{\tilde{K}-1}{\tilde{K}\Delta f}$ and $f = 0, \frac{1}{N_rT_f}, \cdots, \frac{N_r-1}{N_rT_f}$.
The corresponding stacked candidate matrices form DFT matrices. 
Accordingly, the 2D-MUSIC spectrum for angle $\theta_\ell$ can be efficiently evaluated using a 2D-FFT, and we can enhance the resolution of the spectrum by applying zero-padding. 
This method is referred to as the \textit{AoA-Based 2D-FFT-MUSIC} (AB2FM) algorithm, which is summarized in Algorithm \ref{alg:FFT2-MUSIC}.

Notably, the estimated AoAs can be refined by searching within a narrow range around pre-estimated values during DDE\&CS.
This joint estimation approach enables simultaneous estimation of all parameters, eliminating the need for separate matching procedures.
Furthermore, the estimated Doppler can be effectively utilized to distinguish between static and dynamic paths, as the AoAs estimated in AES encompass both yet remain indistinguishable.

\begin{algorithm}[tb!]
    \caption{Proposed AB2FM algorithm}
    \begin{algorithmic}[1]
        \STATE \textbf{Input:} \added{$T_f$, $\Delta f$, $M_c$, $\mathbf{W}_s$, $L$, $\tilde{K}$, $G_D$, $G_\tau$, estimated AoAs vector $\tilde{\Theta}$ from Algorithm \ref{alg:AoA_LoS}};
        
        \STATE \textbf{Initialization:} Set the estimated Doppler vector $\mathbf{f}$ and delay vectors $\boldsymbol{\tau}$ as zero vectors of length $(L-1)$;
        
        \STATE Obtain $\boldsymbol{\Xi}_{\tilde{k}}[m]$ for $m=0, \cdots, M_c-1$ and $\tilde{k} = 0, \cdots, \tilde{K}$ according to (\ref{equ:stacked_measurement_matrix});
        
        \STATE Calculate $\overline{\mathbf{R}}_{\overline{\boldsymbol{\xi}}}$ using (\ref{equ:vectorized_measurement_matrix}), (\ref{equ:mean_vectorized_measurement_matrix}) and (\ref{equ:covariance_matrix_doppler_delay_before_smoothing});
        
        \STATE Compute the eigenvector matrix $\mathbf{U}_{\text{D}, s}$ corresponding to the $L$ largest eigenvalues of $\overline{\mathbf{R}}_{\overline{\boldsymbol{\xi}}}$ using iterative methods;
        
        \FOR{$\ell=1:L-1$}
        
            \STATE Initialize temporary variables $\mathbf{P}$ and $\mathbf{Q}$ as $G_D\times G_\tau$ zero matrices;
        
            \STATE Set $\theta_\ell = \tilde{\Theta}(\ell)$ and construct $\mathbf{a}(\theta_\ell)$;
        
            \FOR{$i = 0:L-1$}
        
                \STATE $\mathbf{u}_i = \mathbf{U}_{\text{D}, s}(:, i)$, according to (\ref{equ:Space_Partition});
                
                \STATE Compute $\tilde{\mathbf{u}}_{i, \theta_\ell} = \left[\mathbf{W}_s^H\mathbf{a}(\theta_\ell)\otimes \mathbf{1}_{\tilde{K}\times 1}\right]\odot \mathbf{u}_i$;
        
                \STATE Reshape $\tilde{\mathbf{u}}_{i, \theta_\ell}$ by $\tilde{\mathbf{U}}_{i, \theta_\ell} = \operatorname{reshape}(\tilde{\mathbf{u}}_{i, \theta_\ell}, \tilde{K}, N_r)$;
        
                \STATE Update $\mathbf{Q} = \mathbf{Q} + \operatorname{FFT2}(\tilde{\mathbf{U}}_{i, \theta_\ell}, G_{\tau}, G_D)$ according to (\ref{equ:2D_FFT_derivation});
                
            \ENDFOR
        
            \STATE Compute $\mathbf{P} = \frac{1}{N_r\tilde{K}-\mathbf{Q}}$ according to (\ref{equ:MUSIC_Spectrum_Signal_Subspace});
        
            \STATE Search for the maximum value of $\mathbf{P}$ and obtain the corresponding indices $g_{\tau}$ and $g_D$;
        
            \STATE Update $\boldsymbol{\tau}(\ell) = \frac{g_{\tau}}{G_\tau\Delta f}$ and $\mathbf{f}(\ell) = \frac{g_D}{G_DT_f} - \frac{1}{2T_f}$;
        
        \ENDFOR
        
        \STATE \textbf{Output:} AoA-Doppler-delay pairs $\left(\tilde{\Theta}_{\setminus \tilde{\theta}_0}, \mathbf{f}, \boldsymbol{\tau}\right)$.
    \end{algorithmic}
    \label{alg:FFT2-MUSIC}
\end{algorithm}
\subsection{Complexity Analysis}
\label{subsec:complexity_analysis}

\added{
    For Algorithm \ref{alg:AoA_LoS}, the EVD for the signal subspace, extracted using the iterative method, has a complexity of $\mathcal{O}(N_rL^2)$, as only $L$ eigenvectors need to be computed.
    The MUSIC spectrum analysis introduces an additional complexity of approximately $\mathcal{O}(LG_{\theta}\log(G_{\theta}))$.
    Consequently, the overall computational complexity of Algorithm \ref{alg:AoA_LoS} is $\mathcal{O}(N_rL^2+LG_{\theta}\log (G_{\theta}))$.
}

\added{
    Simiarly, for Algorithm \ref{alg:FFT2-MUSIC}, the EVD process has a complexity of $\mathcal{O}(N_r\tilde{K}L^2)$.
    Within the inner loop, the update process, which is based on the $\operatorname{FFT2}$ operation, is executed $L\times(L-1)$ times, with each iteration incurring a complexity of $\mathcal{O} \left(G_DG_{\tau}\log (G_DG_{\tau})\right)$.
    Additionally, the element-wise operation in (\ref{equ:MUSIC_Spectrum_Signal_Subspace}) has a negligible complexity of $\mathcal{O}(N_r\tilde{K})$.
    Thus, the overall complexity of Algorithm \ref{alg:FFT2-MUSIC} for MUSIC spectrum analysis is approximately $\mathcal{O}\left(N_r\tilde{K}L^2 + L^2G_DG_{\tau}\log(G_DG_{\tau})\right)$.}
In contrast, the traditional 2D-MUSIC algorithm requires an EVD of the noise subspace, computing nearly all the eigenvectors since $L\ll N_r\tilde{K}$, with a complexity of $\mathcal{O}(N_r^3\tilde{K}^3)$.
The 2D-MUSIC spectrum calculation further incurs a complexity of $\mathcal{O}(G_DG_{\tau}N_r^2\tilde{K}^2)$ per search, repeated $(L-1)$ times.
Thus, the overall complexity is $\mathcal{O}\left(N_r^3\tilde{K}^3 + LG_DG_rN_r^2\tilde{K}^2\right)$.
Comparatively, the proposed AB2FM algorithm shows significantly lower complexity due to $\max(N_r, K) \ll \min(G_D, G_\tau)$.

\section{Simulation Results}
\label{sec:simulation_results}
\subsection{Simulation Setup}
\label{subsec:simulation_setup}
In this section, we present numerical results to illustrate the performance of the proposed algorithm. 
By default, the number of antennas for the receiver is set to $N_r=16$, with half-wavelength spacing. 
The frequency of the first subcarrier is $f_0 = 26\text{GHz}$ within the mmWave frequency band. 
Without loss of generality, the bandwidth is set to $B=100\text{MHz}$, containing $K = 32$ subcarriers with a frequency interval of $\Delta f = 3.125\text{MHz}$. 
The OFDM symbol period is thus $T_{s} = 1/\Delta f = 0.32\text{us}$. 
Each frame in the AES and DDE\&CS consists of $M_s = N_r = 16$ BS-TSs (1 symbol for each BS-TS is assumed, i.e., $T_b = T_s$) and $M_c = 32$ OFDM symbols, respectively, with the interval between two consecutive frames fixed at $M_f = 1024$ OFDM symbols, whose corresponding time interval are $T_{sf} = M_sT_b = 5.12\text{us}$, $T_{cf} = M_cT_s = 10.24\text{us}$ and $T_f = 327.68\text{us}$. 

In each scenario, the AoAs are assumed to be uniformly distributed within the range of $[-80^{\circ}, 80^{\circ}]$, with a minimal interval of $5^{\circ}$. 
\added{The number of multipaths, $L$, is constrained to a maximum of 5 due to the inherent sparsity of the mmWave channel.}
Three types of multipath components are considered: static LoS path, static NLoS path, and dynamic NLoS path. 
\added{To differentiate them in the simulation, Doppler shifts are applied exclusively to dynamic paths.}
\added{The LoS path is assumed to exhibit a power level that exceeds those of the NLoS paths by at least 3 dB.}
\added{The delays associated with all paths are uniformly distributed within $(0, \tau_{\text{max}})$, with the LoS path having the smallest delay.}
\added{Meanwhile, the Doppler shifts of dynamic NLoS paths follow a uniform distribution over $(-f_{\text{max}}, f_{\text{max}})$, where $\tau_{\text{max}} = T_{s}$ and $f_{\text{max}} = 1/2T_f$.}
The corresponding distance and velocity are $d_{\text{max}} = 96\text{m}$ and $v_{\text{max}} = 17.6\text{m/s}$, respectively, which are also termed as maximum detectable distance and velocity.

We employ the normalized mean-squared error (NMSE), the complementary cumulative distribution function (CCDF) and the detection accuracy as key performance indicators.

The NMSE is defined as:
\begin{equation}
    \text{NMSE}_{a} = \frac{E\left\{|\tilde{a}-a|^2\right\}}{E\left\{|a|^2\right\}},
    \label{equ:NMSE}
\end{equation}
where $a$ and $\tilde{a}$ represent the true value and the corresponding estimated value, respectively.

The CCDF is defined as:
\begin{equation}
    C(x) = 1-P(X\leq x),
    \label{equ:CCDF}
\end{equation}
where $X$ represents the variable of interest and $x$ is the threshold. 
This allows for a clear comparison of performance differences at specific error thresholds. 
A lower CCDF curve indicates higher reliability as it reflects a reduced probability of estimation errors exceeding the threshold.

The detection accuracy is defined as $\eta = \frac{N_{\text{correct}}}{N}$, where $N_{\text{correct}}$ and $N$ represent the number of correct detections and the total number of detections, respectively. 
A detection is considered correct if $|\tilde{a}-a|\leq \epsilon$, where $\epsilon$ denotes the predetermined threshold.

\subsection{Performance of AoA Estimation}
\label{subsec:performance_aoa}
\begin{figure}[tb!]
    \centering
    \includegraphics[width=\linewidth]{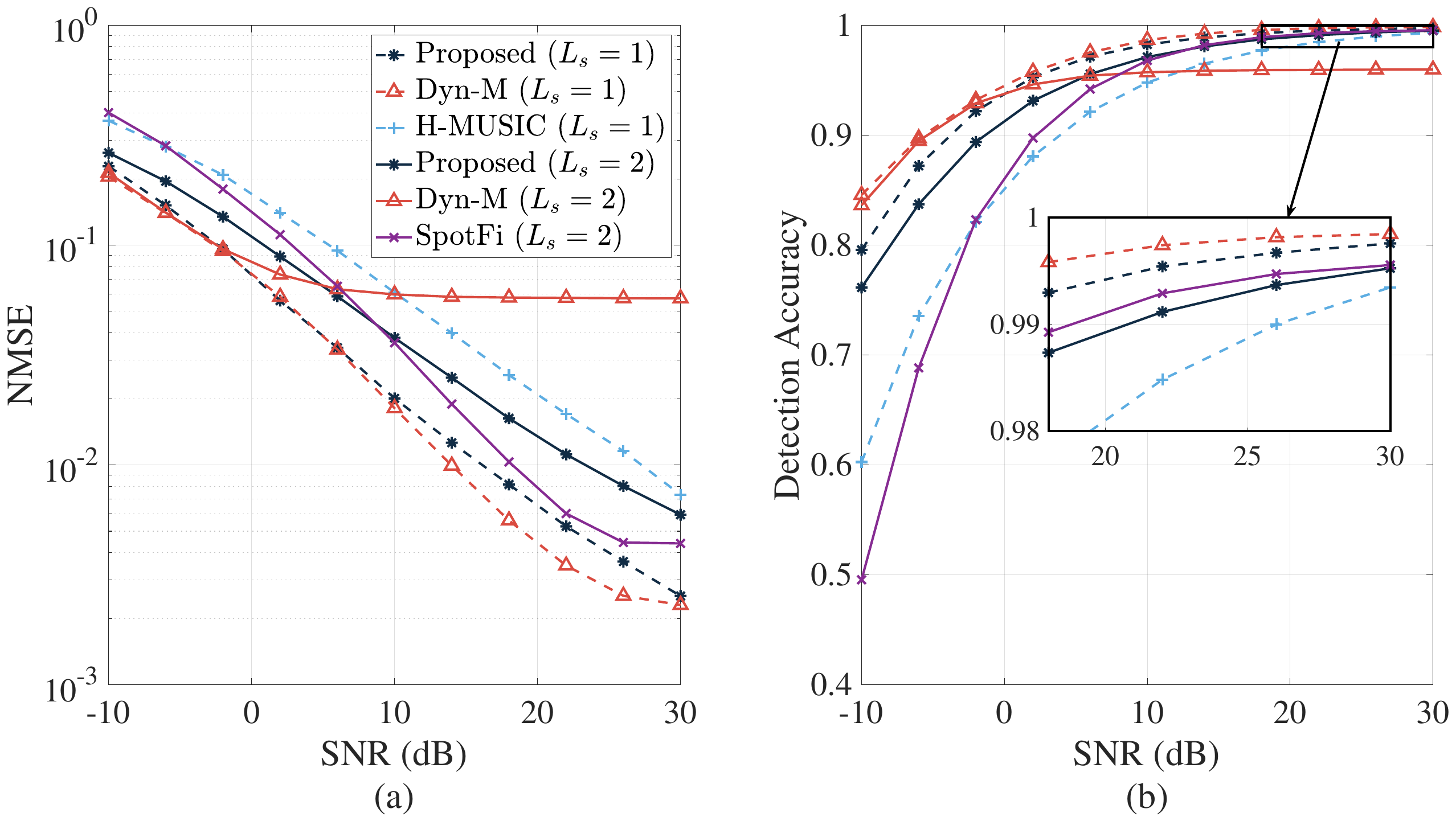}
    \caption{(a) NMSE and (b) detection accuracy of AoA estimation, with a threshold set as $\epsilon_{\angle} = 3^\circ$. Note: These figures share the same legend.}
    \label{fig:AoA_NMSE}
\end{figure}

In this subsection, we evaluate the AoA estimation performance of the proposed frequency smoothed MUSIC algorithm for hybrid arrays, compared to H-MUSIC \cite{chuang2015high}, SpotFi \cite{kotaru2015spotfi} and Dyn-MUSIC \cite{li2016dynamic-MUSIC} algorithms, where H-MUSIC is designed for hybrid arrays as well. 
These algorithms differ in their ability to detect static and dynamic paths: H-MUSIC and Dyn-MUSIC can estimate merged static paths (coherent signals) and all dynamic paths, whereas the proposed algorithm and SpotFi can estimate all static and dynamic paths separately. 

Fig. \ref{fig:AoA_NMSE} shows the NMSE (a) and detection accuracy (b) of AoA estimation for different algorithms, with a detection accuracy threshold of $\epsilon_{\angle} = 3^\circ$. 
We set $L_s = 1$ and $L_d = 3$ to evaluate the proposed algorithm, Dyn-MUSIC, and H-MUSIC in scenarios without coherent signals. 
The results indicate that the proposed algorithm and Dyn-MUSIC exhibit similar NMSE performance, even at low SNR, with both outperforming H-MUSIC. 

To evaluate the capability of handling coherent signals, $L_s$ is set to 2 and $L_d$ to 3 for the proposed algorithm, Dyn-MUSIC, and SpotFi. 
For fairness in NMSE calculation, Dyn-MUSIC considered merged static and dynamic paths, while our algorithm and SpotFi included all paths.
At low SNR, the proposed algorithm and Dyn-MUSIC outperform SpotFi in terms of both NMSE and detection accuracy. 
However, at high SNR, the NMSE for Dyn-MUSIC reaches a plateau because the estimated AoA of the merged path deviates from the LoS path, thereby limiting further NMSE reduction.
Furthermore, SpotFi shows slightly better performance than the proposed algorithm.
Notably, although both Dyn-MUSIC and SpotFi are designed for digital arrays, the proposed algorithm, tailored for hybrid arrays, demonstrates comparable performance to these digital array-based algorithms.

\subsection{Performance of Beam Designs}
\label{subsec:performance_beam}
\begin{figure}[tb!]
    \centering
    \includegraphics[width=\linewidth]{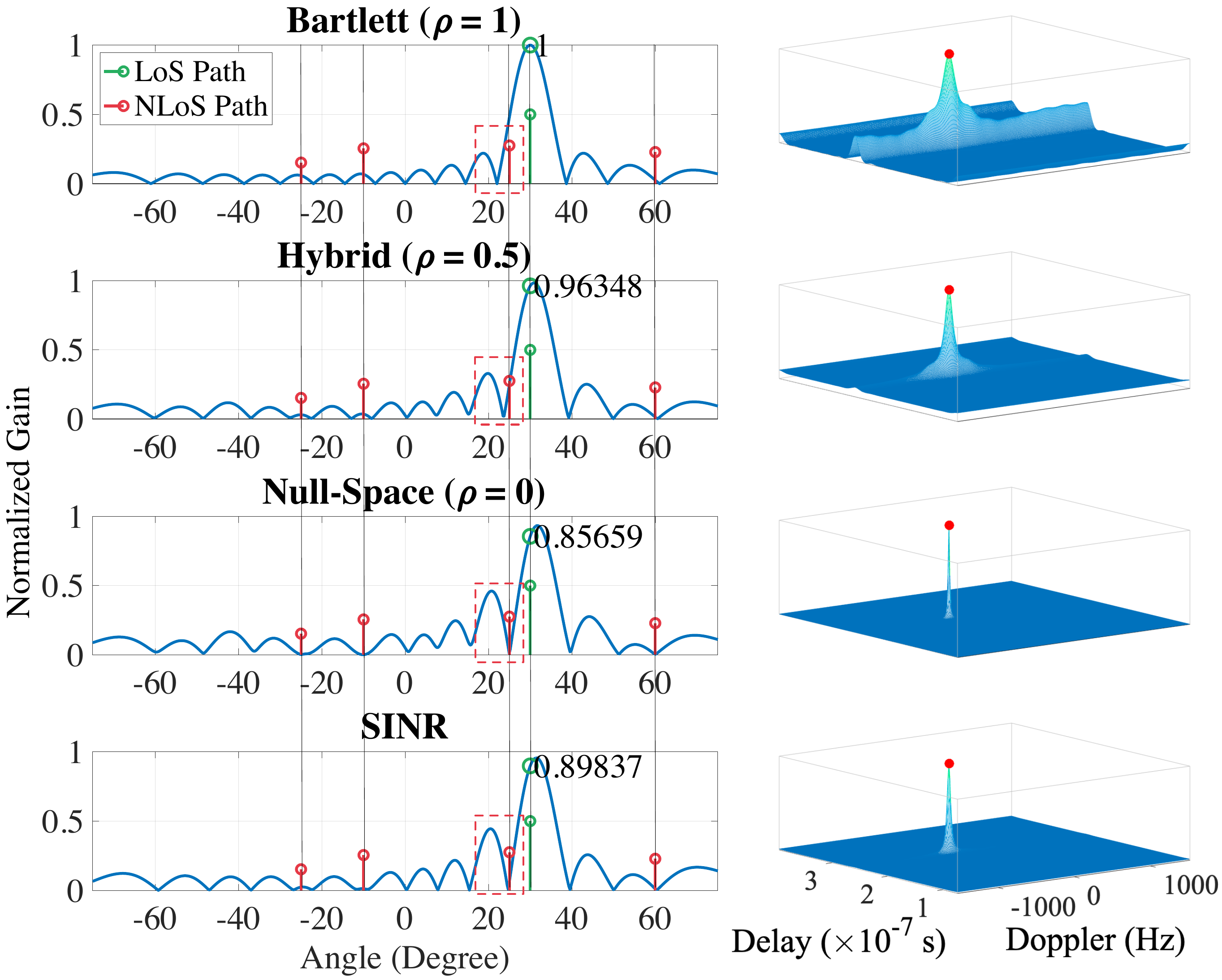}
    \caption{An example of array response (left) and Doppler-delay MUSIC spectrum (right) corresponding to various BF vector design methods under SNR = 0dB: Bartlett beam (top), hybrid beam ($\rho=0.5$, second), NS beam (third), and SINR beam (bottom). It is assumed that there are five paths with angles \{$-25^{\circ}, -10^{\circ}, 25^{\circ}, 30^{\circ}$(LoS)$, 60^{\circ}$\}, with the LoS path having the highest power and the SNR = $-$5dB.}
    \label{fig:array_response_MUSIC_Spectrum}
\end{figure}
To compare the impact of reference signals received by different communication beams, as discussed in Section \ref{subsec:proposed_combining_vector_for_communication}, we conducted experiments using these beams to estimate Doppler and delay. 

Fig. \ref{fig:array_response_MUSIC_Spectrum} depicts the array response in the spatial domain for beams designed by different approaches, along with the Doppler-delay MUSIC spectrums at an AoA of $25^\circ$ (dynamic path) using the corresponding beams. 
The Bartlett beam achieves the highest gain for the LoS path, while the NS beam achieves the lowest. 
The red-dashed rectangles highlight the NS beam's superior suppression of the NLoS paths, with the hybrid and SINR beams striking a balance between NLoS suppression and LoS gain. 
In the Doppler-axis of the spectrum, using the Bartlett beam as a reference, the peak is broadened by NLoS phase interference.  
As $\rho$ decreases, the peak sharpens, with the NS beam yielding the sharpest peak. 
Both the hybrid and SINR beams demonstrate notably sharp peaks.

\begin{figure}[tb!]
    \centering
    \includegraphics[width=\linewidth]{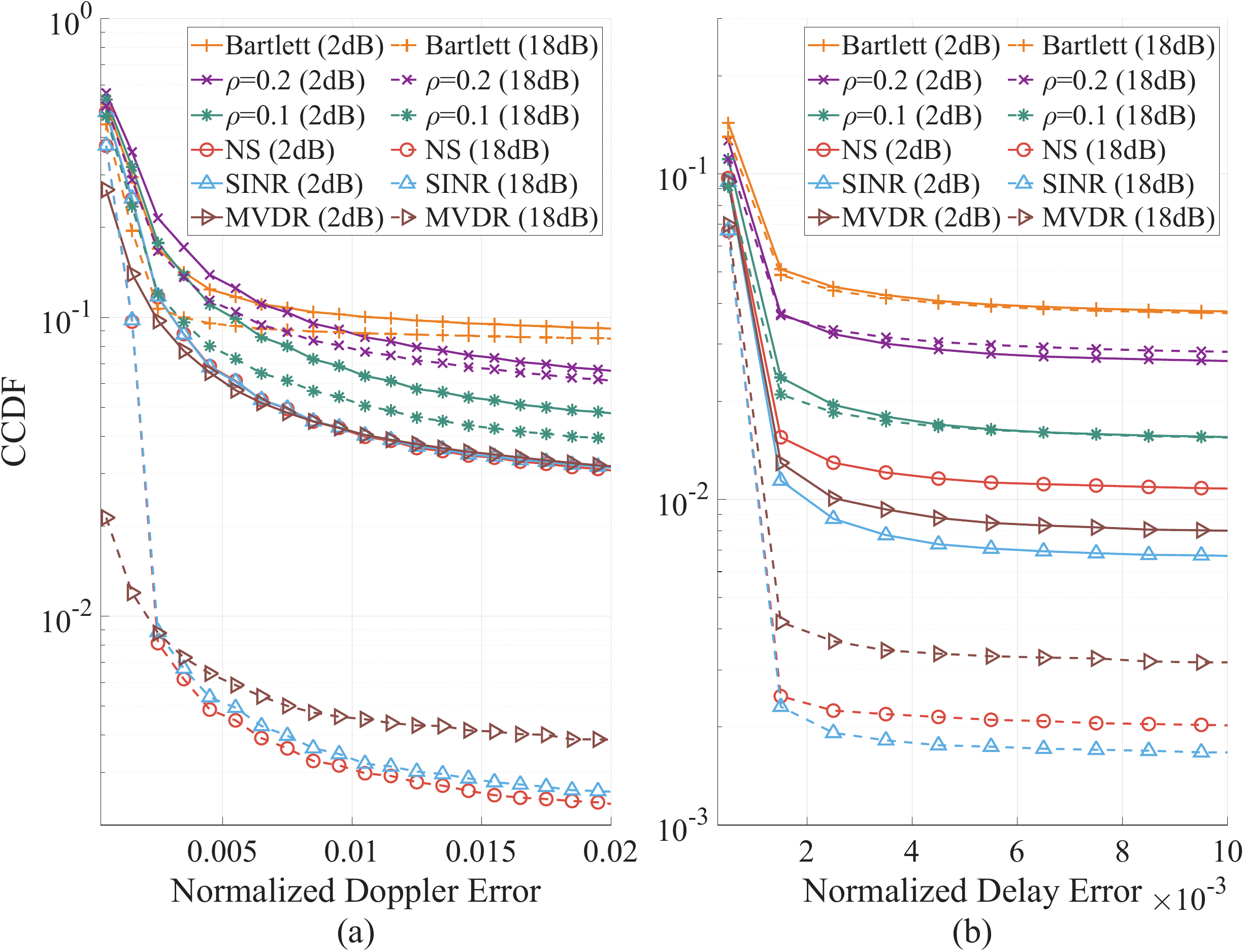}
    \caption{CCDF versus (a) normalized Doppler error and (b) normalized delay error for different proposed beams.}
    \label{fig:CCDF_CombVec}
\end{figure}
Fig. \ref{fig:CCDF_CombVec} illustrates the CCDF curves against (a) normalized Doppler error and (b) normalized delay error for different communication beams. 
The number of static and dynamic paths is set to $L_s = 2$ and $L_d=3$.
The solid and dashed lines represent simulation conditions under SNR = 2dB and SNR = 18dB, respectively. 
Different SNR levels impact the optimization results of the beams. 
Once designed, the signals from these beams are used as references to estimate Doppler and delay under the same SNR conditions. 
\added{While MVDR was not discussed earlier, we include it in the simulations as a classical BF method to further demonstrate the applicability of the proposed CFO and TO elimination scheme.}
As shown, the NS and SINR beams achieve higher accuracy in Doppler and delay estimation than other beams. 
\added{As $\rho$ increases, accuracy declines due to significant interference from NLoS paths, which outweighs the increase in LoS path gain.}
\added{The NS beam outperforms others in Doppler estimation, while the SINR beam performs better in delay estimation, suggesting that Doppler estimation is more sensitive to NLoS-induced phase interference.}
\added{Furthermore, Doppler estimation benefits from ``cleaner'' reference signals, making the NS beam the optimal choice for this purpose.}
\added{Notably, the MVDR beam achieves performance comparable to both the NS and SINR beams in Doppler and delay estimation.}

\begin{figure}[tb!]
    \centering
    \includegraphics[width=\linewidth]{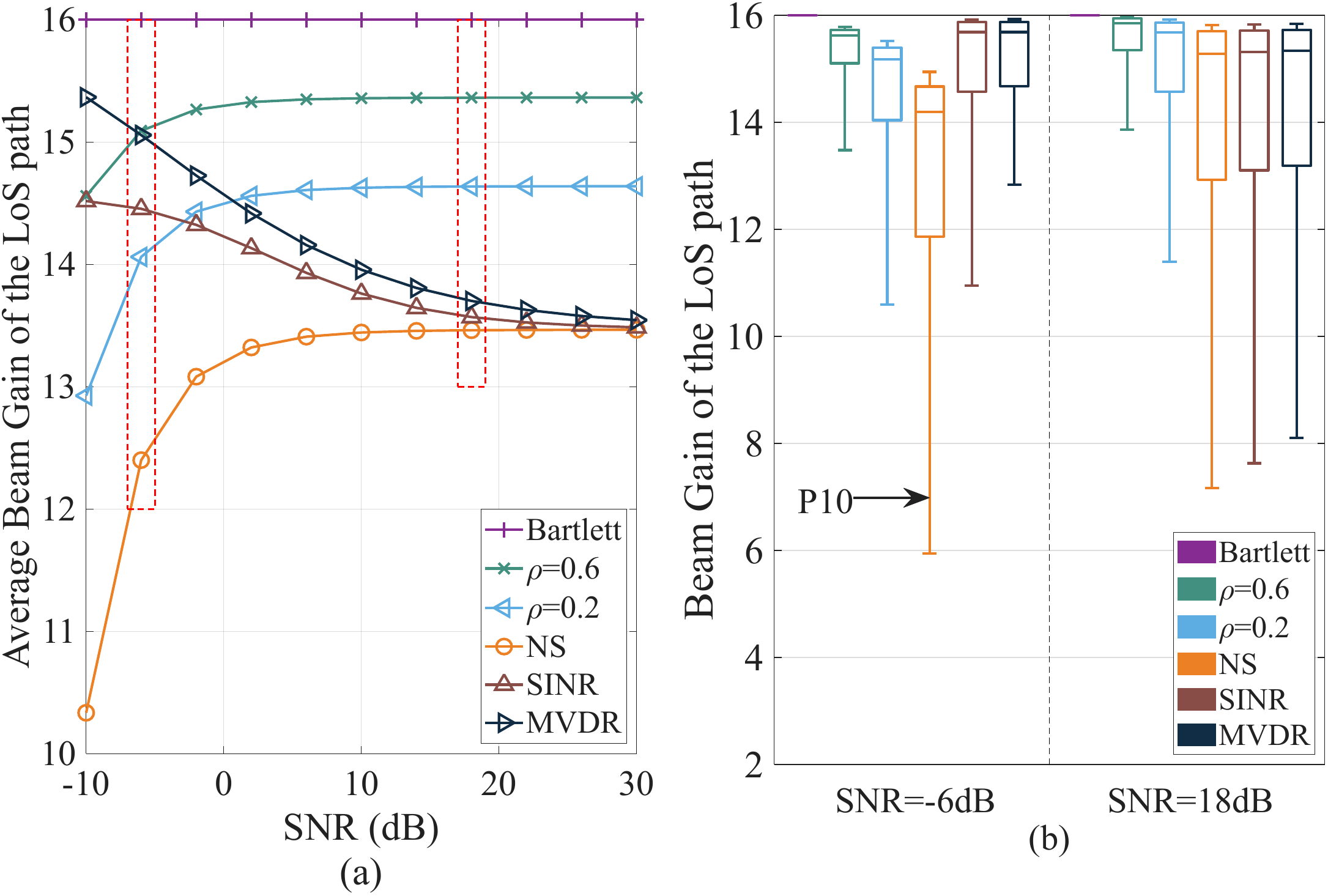}
    \caption{Comparison of (a) average beam gain of LoS path across various SNR levels and beams, and (b) beam gain distribution of LoS path for beams at SNR levels of -6dB and 18dB.}
    \label{fig:Beam_Gain_CombVec}
\end{figure}
In addition to the sensing performance, we use the beam gain of the LoS path as an indicator to briefly evaluate the communication performance of the beam, as Fig. \ref{fig:Beam_Gain_CombVec} shown. 
Fig. \ref{fig:Beam_Gain_CombVec} (a) illustrates the average beam gain toward the LoS path with respect to various SNR levels and beams. 
It is evident that the NS beam consistently has the lowest average beam gain, while the Bartlett beam has the highest, under the same SNR. 
The average beam gain increases with $\rho$. 
\added{Interestingly, the SINR and MVDR beam's average beam gain decreases as the SNR increases, which is counter-intuitive.}
This can be explained by considering that at low SNR, the noise power dominates the optimization problem, leading to an emphasis on amplifying the LoS path. 
At higher SNR, interference from NLoS paths becomes more significant than noise, leading the optimization process to prioritize the suppression of NLoS paths.
\added{As a result, the performance of the SINR and MVDR beam approaches that of the NS beam.}

Fig. \ref{fig:Beam_Gain_CombVec} (b) presents the beam gain distribution toward the LoS path across various numerical simulations. 
The whiskers extend from the edges of the box, which represent 25th- and 75th-percentiles, to the 10th- and 90th-percentiles.
Specifically, we selected two typical SNR levels, marked by the red dashed rectangles in Fig. \ref{fig:Beam_Gain_CombVec} (a): SNR = $-$6dB (left part) and 18dB (right part). 
\added{At a low SNR, the SINR and MVDR beam outperform the NS beam, with its 25th-percentile gain nearly matching the 75th-percentile gain of the NS beam.}
Conversely, at high SNR, the performance of these three beams converges, as previously discussed. 
Notably, the 10th-percentile gain of the NS beam is significantly less than half that of the Bartlett beam, indicating that in over 10\% of cases, the beam gain would severely impact communication performance.
However, as $\rho$ increases, both the median beam gain and the 10th-percentile beam gain improve significantly, and the gain disparities are greatly reduced.
Therefore, in practice, we can choose different beams based on the situation. 
If we opt for the hybrid beam, choosing an appropriate $\rho$ is crucial.

\subsection{Performance of Doppler Estimation}
\label{subsec:performance_Doppler}
\begin{figure*}[tb!]
    \centering
    \includegraphics[width=\linewidth]{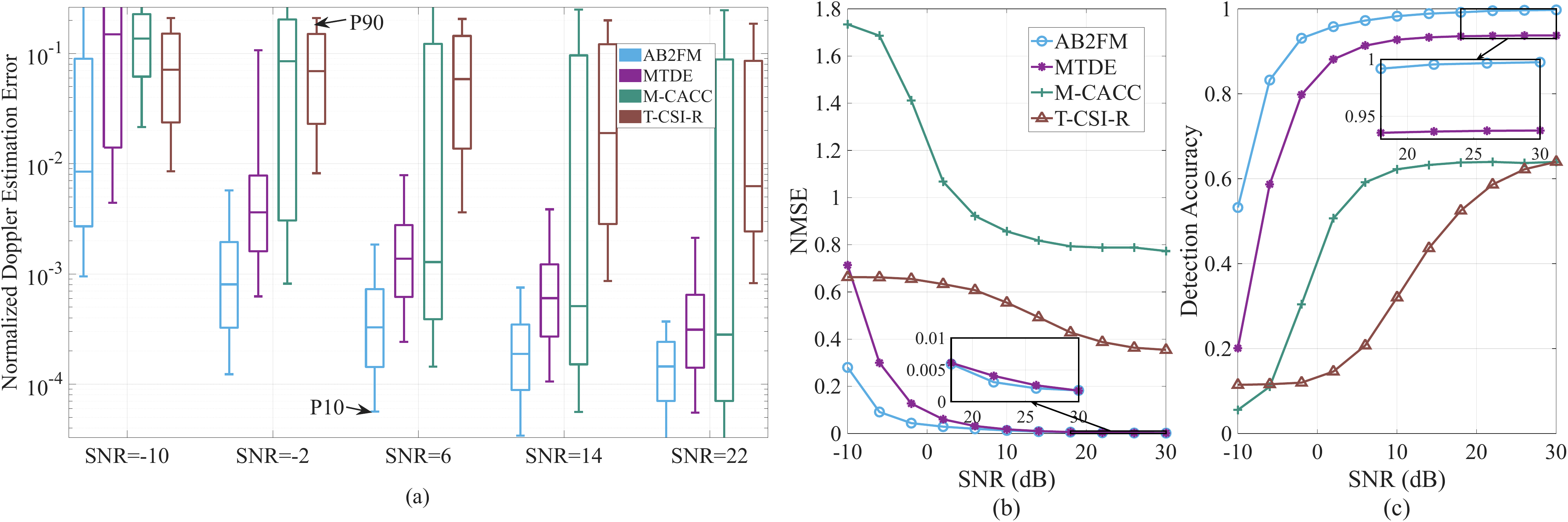}
    \caption{(a) Box plot of Doppler estimation errors normalized by the maximum frequency range, with whisker extending to the 10th- and 90th-percentiles (P10 and P90); (b) NMSE of Doppler estimation; (c) Detection accuracy of Doppler estimation, with a threshold set at 1\textperthousand{} of the maximum detectable frequency range. \added{Note: The same line types are used in (b) and (c), with the legend included in (b).}}
    \label{fig:Box_NMSE_Doppler}
\end{figure*}

To evaluate the effectiveness of the AB2FM algorithm in Doppler estimation, we conducted a comparison with three other algorithms: the MTDE estimation algorithm based on nullspace techniques \cite{zhao2023multiple}, the CACC based mirror-MUSIC algorithm (M-CACC) \cite{ni2021uplink-CACC}, and CSI-R based Taylor-series-MUSIC algorithm (T-CSI-R) \cite{ni2023uplink-CSIR}.
These algorithms are specifically designed for sensing multiple targets in asynchronous systems with digital arrays.

The Doppler estimation performance of various algorithms is compared in Fig. \ref{fig:Box_NMSE_Doppler} (a)-(c). 
\added{All algorithms are evaluated under the conditions of $L_s=2$ and $L_d=3$.} 
\added{Additionally, random CFO and TO are introduced to the signals for all these algorithms, ensuring that the results reflect realistic performance under diverse channel conditions.}

Fig. \ref{fig:Box_NMSE_Doppler} (a) shows the normalized Doppler estimation error distribution across different SNR levels for each algorithm. 
The results demonstrate that AB2FM consistently outperforms the other algorithms.
At each SNR level, AB2FM exhibits a lower median error rate, with its 90th-percentiles errors falling below the 75-percentiles errors of MTDE.
Furthermore, both AB2FM and MTDE show smaller interquartile ranges (IQRs) compared to M-CACC and T-CSI-R, indicating more consistent performance.
\added{The median error for both MTDE and M-CACC decreases rapidly as SNR increases, which suggests that these algorithms are more sensitive to noise.}
\added{Additionally, the IQR fails to converge in M-CACC due to its high Doppler and delay mismatch rate, and in T-CSI-R due to the convergence issue of the Taylor series expansion.}

Fig. \ref{fig:Box_NMSE_Doppler} (b) and (c) present the NMSE and detection accuracy results of Doppler estimation. 
The detection accuracy is computed using a threshold set at 1\textperthousand{} of $2f_{\text{max}}$, meaning an estimate is accurate if it deviates by less than $\epsilon_v = 0.0352\text{m/s}$ under this parameter setup. 
AB2FM outperforms other algorithms in NMSE, even at low SNR, which aligns with the box plot results. 
Specifically, AB2FM demonstrates robust NMSE performance across varying SNR levels. 
\added{When the SNR exceeds 0dB, AB2FM achieves over 90\% detection accuracy, and MTDE also performs well at high SNR.}
\added{In contrast, M-CACC and T-CSI-R remain limited to around 65\% accuracy even at high SNR.}

\subsection{Performance of Delay Estimation}
\label{subsec:performance_delay}
\begin{figure*}[tb!]
    \centering
    \includegraphics[width=\linewidth]{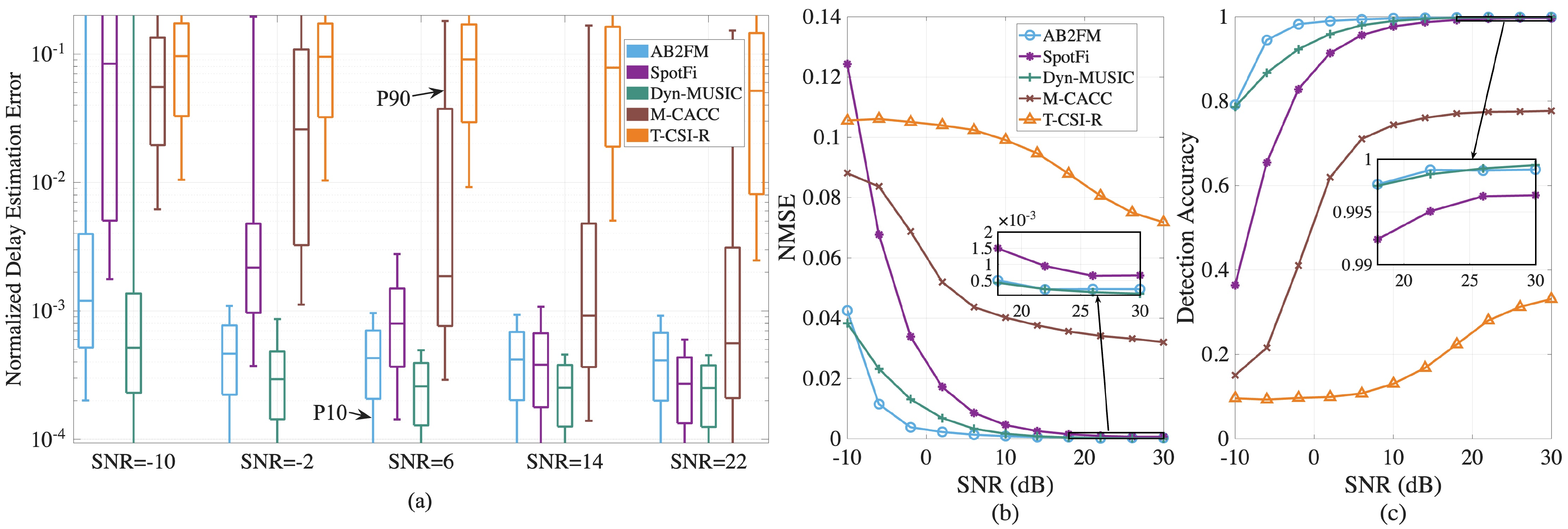}
    \caption{(a) Box plot of delay estimation errors normalized by the maximum delay range, with whisker extending to the 10th- and 90th-percentiles (P10 and P90); (b) NMSE of delay estimation; (c) Detection accuracy of delay estimation, with a threshold set at 1\textperthousand{} of the maximum detectable delay range. \added{Note: The same line types are used in (b) and (c), with the legend included in (b).}}
    \label{fig:Box_NMSE_Delay}
\end{figure*}
Fig. \ref{fig:Box_NMSE_Delay} (a)-(c) illustrate the performance of delay estimation.
\added{We compare AB2FM, SpotFi, Dyn-MUSIC, M-CACC and T-CSI-R under consistent settings as used for Doppler estimation.}
Random CFO and TO are incorporated into AB2FM, M-CACC, and T-CSI-R, whereas SpotFi and Dyn-MUSIC serve as benchmarks free from CFO and TO.
The delay referred to here is relative to the LoS path. 
Fig. \ref{fig:Box_NMSE_Delay} (a) presents the box plot of normalized delay estimation error versus SNR for the various algorithms. 
At low SNR, AB2FM and Dyn-MUSIC exhibit low estimation errors, with median values around $5\times 10^{-4}$.
However, the IQR for Dyn-MUSIC, which spans from P10 to P90, is significantly broader than that of AB2FM, indicating lower detection accuracy for Dyn-MUSIC at low SNR.
Meanwhile, SpotFi, M-CACC, and T-CSI-R exhibit poor performance, indicated by their large IQRs and high median values.
\added{As SNR increases, SpotFi's performance nearly matches Dyn-MUSIC and slightly outperforms AB2FM at high SNR.}
\added{Notably, both SpotFi and Dyn-MUSIC utilize digital array, while AB2FM, based on hybrid array, still provides comparable performance.}

Fig. \ref{fig:Box_NMSE_Delay} (b) and (c) illustrate the NMSE and detection accuracy for delay estimation. 
The detection accuracy is calculated using a threshold set at 1\textperthousand{} of $T_{s}$, corresponding to a distance deviation of $\epsilon_d = 0.096\text{m}$.
From Fig. \ref{fig:Box_NMSE_Delay} (b), it is evident that AB2FM outperforms other algorithms in both NMSE and detection accuracy under low SNR conditions.
As the SNR increases, AB2FM maintains its competitive edge in NMSE, matching the performance of SpotFi and Dyn-MUSIC. 
Although M-CACC and T-CSI-R show decreasing NMSE with increasing SNR, their overall performance lags behind that of AB2FM, SpotFi, and Dyn-MUSIC. 
Specifically, as Fig. \ref{fig:Box_NMSE_Delay} (c) shows, AB2FM achieves the highest detection accuracy at low SNR and reaches over 90\% even when the SNR is lower than 0dB. 
With increasing SNR, the detection accuracies of AB2FM, SpotFi, and Dyn-MUSIC all converge towards 1, surpassing 99\% accuracy. 

Overall, the proposed AB2FM algorithm for hybrid array systems excels in both Doppler and delay estimation, maintaining robustness in low SNR, and matches or exceeds the performance of algorithms developed for digital array systems.

\section{Conclusion}
\label{sec:conclusion}
We proposed an ISAC framework specifically designed for uplink sensing in PMNs with mmWave hybrid arrays, including beam scanning, frame structures, BF designs, and parameter estimation algorithms. 
This framework can be well integrated into existing communication systems, effectively addressing clock asynchrony while enabling parameter estimation. 
We also developed a low-complexity AB2FM algorithm to enhance the framework's practicality, ensuring efficient processing and ease of implementation. 
Simulation results demonstrate that our framework performs comparably to, or even outperforms, other high-resolution estimation algorithms developed for fully digital arrays. 
\added{Our work can be potentially extended by incorporating a constant modulus constraint in BF vector design to enhance practical applicability.
Additionally, exploring advanced waveform designs (e.g., orthogonal time frequency space (OTFS)) and extending the framework to multiple UE scenarios are promising directions for future research.}


%

\appendices
\section{Derivation of $\operatorname{rank}(\mathbf{R}_{\boldsymbol{\varsigma}_k})$ and the Combined Steering Vector for Static Paths}
\label{sec:appendix_a}
\added{
    In our model, the covariance matrix of $\boldsymbol{\varsigma}_k$ is given by (\ref{equ:covariance_matrix_AoA_k_subcarrier}), where $\tilde{\mathbf{H}}_k$ is defined in (\ref{equ:H_k_weighted_average}).
    Since our focus is on analyzing the rank of $\tilde{\mathbf{H}}_k$, and $|\eta_k|^2$ is a scalar weight that does not affect the rank, we simplify $\tilde{\mathbf{H}}_k$ as
    \begin{equation}
        \tilde{\mathbf{H}}_k = \mathbb{E}\left[\tilde{\mathbf{h}}_{k}\left(\tilde{\mathbf{h}}_{k}\right)^H\right],
        \label{equ:H_k_average} 
    \end{equation}
    where $\tilde{\mathbf{h}}_k$ (defined in (\ref{equ:CSI_discrete_k_subcarrier})) consists of $L = L_s + L_d$ multipath components.
    Since all static components originate from the same UE and share the same transmitted signal, they are coherent.
    This means there exist complex scalars $\mu_{k, \ell}$ (with $\mu_{k, 0} = 1$) such that
    \begin{equation}
        (\tilde{\mathbf{h}}_{k})_\ell=\mu_{k, \ell}(\tilde{\mathbf{h}}_{k})_{0},\forall\ell=1,\cdots, L_s-1.
        \label{equ:h_equality}
    \end{equation}
    By defining
    \begin{equation}
        \boldsymbol{\mu}_{k, s} = \left[1, \mu_{k, 1}, \cdots, \mu_{k, L_s-1}\right]^T\in \mathbb{C}^{L_s},
    \end{equation}
    the channel vector for the $k$-th subcarrier can be partitioned as
    \begin{equation}
        \tilde{\mathbf{h}}_{k} = \left[
            \begin{array}{c}
                \boldsymbol{\mu}_{k, s} (\tilde{\mathbf{h}}_{k})_0\\
                \tilde{\mathbf{h}}_{k, d}
            \end{array}
        \right],
    \end{equation}
    where $\tilde{\mathbf{h}}_{k, d}\in \mathbb{C}^{L_d}$ collects the coefficients of the dynamic paths, which are assumed to be mutually incoherent (each characterized by a unique Doppler shift as discussed in Section \ref{subsec:scenario_and_received_signal}).
    Therefore, the rank of (\ref{equ:H_k_average}) is $\operatorname{rank}(\tilde{\mathbf{H}}_k) = L_d+1$. 
    Correspondingly, neglecting the noise term, the covariance of the stacked measurements in (\ref{equ:covariance_matrix_AoA_k_subcarrier}) is 
    \begin{equation}
        \mathbf{R}_{\boldsymbol{\varsigma}_k} = \mathbf{W}_s^H\mathbf{A}
        \tilde{\mathbf{H}}_k
        \mathbf{A}^H\mathbf{W}_s.
    \end{equation}
    Since $\mathbf{A}$ (defined in (\ref{equ:steering_matrix})) of size $N_r\times L$ is assumed to have full column rank (i.e., $\theta_i \neq \theta_j, \forall i \neq j$) and $\mathbf{W}_s$ is unitary, we have
    \begin{equation}
        \operatorname{rank}(\mathbf{R}_{\boldsymbol{\varsigma}_k}) = \operatorname{rank}(\tilde{\mathbf{H}}_k) = L_d + 1.
    \end{equation}
    Moreover, the coherence relationship in (\ref{equ:h_equality}) implies that, after multiplication by $\mathbf{A}$, the static paths merge into a single combined steering vector
    \begin{equation} 
        \mathbf{a}_{k, c} = \sum_{\ell = 0}^{L_s-1}\mu_{k, \ell} \mathbf{a}(\theta_\ell),\,\,\, \text{with}\,\, \mu_{k, 0} = 1. 
        \label{equ:combined_steering_vec} 
    \end{equation} 
    In the presence of a dominant LoS path (i.e., $\mu_{k, 0} \gg \mu_{k, \ell}, \forall \ell \neq 0$), we have $\mathbf{a}_{k, c} \approx \mathbf{a}(\theta_0)$; otherwise, $\mathbf{a}_{k, c}$ will generally deviate from any individual steering vector of static path.
    Note that, since the coefficients $\mu_{k, \ell}$ depend on the subcarrier index $k$, the combined steering vector also varies with $k$.
}

\section{Proof of Proposition \ref{lem:rank_H_after_smoothing}}
\label{sec:appendix_B}
\begin{proof}
    \added{
        As analysis in Appendix \ref{sec:appendix_a}, the rank of $\overline{\mathbf{R}}_{\boldsymbol{\varsigma}}$ (\ref{equ:covariance_matrix_AoA_smoothed}) is equivalent to the rank of $\overline{\mathbf{H}}$.
        By neglecting noise and the scalar weights $|\eta_k|^2$, $\overline{\mathbf{H}}$ can be reformulated as
        \begin{equation}
            \begin{aligned}
                \mathbf{\overline{H}} &= \frac{1}{K}\sum_{k=0}^{K-1} \mathbb{E}\left[\tilde{\mathbf{h}}_{k}\left(\tilde{\mathbf{h}}_{k}\right)^H\right]\\
                &= \frac{1}{K}\sum_{k=0}^{K-1} \mathbb{E}\left[
                \boldsymbol{\Phi}_\tau^k\tilde{\mathbf{h}}_{0}\left(\tilde{\mathbf{h}}_{0}\right)^H\left(\boldsymbol{\Phi}_\tau^k\right)^H\right]\\
                &= \frac{1}{K}\sum_{k=0}^{K-1} 
                \boldsymbol{\Phi}_\tau^k\tilde{\mathbf{H}}_0\left(\boldsymbol{\Phi}_\tau^k\right)^H,
            \end{aligned}
            \label{equ:H_after_smoothing_equal}
        \end{equation}
        where $\mathbf{\Phi}_\tau^k=\operatorname{diag}\left(\left[e^{-j 2 \pi \tau_0 k\Delta f}, \cdots, e^{-j 2 \pi \tau_{L-1} k\Delta f}\right]\right)$ and $\tilde{\mathbf{H}}_0 = \mathbb{E}\left[\tilde{\mathbf{h}}_{0}\left(\tilde{\mathbf{h}}_{0}\right)^H\right]$.
    }
    
    \added{
        By taking the Hermitian square root of $\tilde{\mathbf{H}}_0$, i.e., $\tilde{\mathbf{H}}_0 = \boldsymbol{\Psi}\boldsymbol{\Psi}^H$, 
        we construct the block matrix
        \begin{equation}
            \mathbf{G} = \left[\boldsymbol{\Psi},\, \boldsymbol{\Phi}_\tau\boldsymbol{\Psi},\, \cdots,\, \boldsymbol{\Phi}_\tau^{K-1}\boldsymbol{\Psi}\right].
        \end{equation}
        Then, $\overline{\mathbf{H}} = \frac{1}{K}\mathbf{G}\mathbf{G}^H$, so that $\operatorname{rank}(\overline{\mathbf{H}}_0) = \operatorname{rank}(\mathbf{GG}^H) = \operatorname{rank}(\mathbf{G})$.
        The $i$-th row of $\mathbf{G}$ is of the form 
        \begin{equation}
            \left[\psi_{i, 1}, \cdots, \psi_{i, L}\right]\otimes 
            \left[1, e^{-j2\pi \tau_i\Delta f}, \cdots, e^{-j2\pi \tau_i (K-1) \Delta f} \right].
        \end{equation}
        Ignoring the nonzero coefficients $\psi_{i, j}$ (which do not affect linear independence), each row essentially forms a Vandermonde-like vector $\left[1, e^{-j2\pi \tau_i\Delta f}, \cdots, e^{-j2\pi \tau_i (K-1) \Delta f} \right]$.
        Since the phase term $e^{-j2\pi \tau_i\Delta f}$ is distinct for different delays $\tau_i$, these vectors are linearly independent.
        Therefore, if $K\geq L$, $\mathbf{G}$ has full row rank, i.e., $\operatorname{rank}(\overline{\mathbf{H}}) = \operatorname{rank}(\mathbf{G}) = L$.
    }
    
    \added{
        Finally, $\operatorname{rank}(\overline{\mathbf{R}}_{\boldsymbol{\varsigma}}) = \operatorname{rank}(\mathbf{W}_s^H\mathbf{A}\mathbf{\overline{H}} \mathbf{A}^H\mathbf{W}_s) = L$.
    }
\end{proof}

\ifCLASSOPTIONcaptionsoff
  \newpage
\fi

\bibliographystyle{IEEEtran} 
\bibliography{bibtex/bib/IEEEabrv, bibtex/bib/IEEEreference}

\begin{thebibliography}{10}
\providecommand{\url}[1]{#1}
\csname url@samestyle\endcsname
\providecommand{\newblock}{\relax}
\providecommand{\bibinfo}[2]{#2}
\providecommand{\BIBentrySTDinterwordspacing}{\spaceskip=0pt\relax}
\providecommand{\BIBentryALTinterwordstretchfactor}{4}
\providecommand{\BIBentryALTinterwordspacing}{\spaceskip=\fontdimen2\font plus
\BIBentryALTinterwordstretchfactor\fontdimen3\font minus
  \fontdimen4\font\relax}
\providecommand{\BIBforeignlanguage}[2]{{%
\expandafter\ifx\csname l@#1\endcsname\relax
\typeout{** WARNING: IEEEtran.bst: No hyphenation pattern has been}%
\typeout{** loaded for the language `#1'. Using the pattern for}%
\typeout{** the default language instead.}%
\else
\language=\csname l@#1\endcsname
\fi
#2}}
\providecommand{\BIBdecl}{\relax}
\BIBdecl
\renewcommand{\BIBentryALTinterwordstretchfactor}{4}

\bibitem{zhang2021overview}
J.~A. Zhang \emph{et~al.}, ``An overview of signal processing techniques for
  joint communication and radar sensing,'' \emph{IEEE J. Sel. Top. Signal
  Process.}, vol.~15, no.~6, pp. 1295--1315, 2021.

\bibitem{liu2020joint}
F.~Liu, C.~Masouros, A.~P. Petropulu, H.~Griffiths, and L.~Hanzo, ``Joint radar
  and communication design: Applications, state-of-the-art, and the road
  ahead,'' \emph{IEEE Trans. Commun.}, vol.~68, no.~6, pp. 3834--3862, 2020.

\bibitem{zhang2022integration}
J.~A. Zhang, K.~Wu, X.~Huang, Y.~J. Guo, D.~Zhang, and R.~W. Heath,
  ``Integration of radar sensing into communications with asynchronous
  transceivers,'' \emph{IEEE Commun. Mag.}, vol.~60, no.~11, pp. 106--112,
  2022.

\bibitem{ITUrecommendation2023framework}
ITU, ``Framework and overall objectives of the future development of {IMT} for
  2030 and beyond,'' \emph{Int. Telecommun. Union Rec. (ITU-R)}, 2023.

\bibitem{zhang2017PMNframework}
J.~A. Zhang, A.~Cantoni, X.~Huang, Y.~J. Guo, and R.~W. Heath, ``Framework for
  an innovative perceptive mobile network using joint communication and
  sensing,'' in \emph{2017 IEEE 85th VTC Spring}.\hskip 1em plus 0.5em minus
  0.4em\relax IEEE, 2017, pp. 1--5.

\bibitem{rahman2019PMN}
M.~L. Rahman, J.~A. Zhang, X.~Huang, Y.~J. Guo, and R.~W. Heath, ``Framework
  for a perceptive mobile network using joint communication and radar
  sensing,'' \emph{IEEE Trans. Aerosp. Electron. Syst.}, vol.~56, no.~3, pp.
  1926--1941, 2019.

\bibitem{zhang2020PMN}
A.~Zhang, M.~L. Rahman, X.~Huang, Y.~J. Guo, S.~Chen, and R.~W. Heath,
  ``Perceptive mobile networks: Cellular networks with radio vision via joint
  communication and radar sensing,'' \emph{IEEE Veh. Technol. Mag.}, vol.~16,
  no.~2, pp. 20--30, 2020.

\bibitem{sohrabi2016hybrid}
F.~Sohrabi and W.~Yu, ``Hybrid digital and analog beamforming design for
  large-scale antenna arrays,'' \emph{IEEE J. Sel. Top. Signal Process.},
  vol.~10, no.~3, pp. 501--513, 2016.

\bibitem{huang2010hybrid}
X.~Huang, Y.~J. Guo, and J.~D. Bunton, ``A hybrid adaptive antenna array,''
  \emph{IEEE Trans. Wireless Commun.}, vol.~9, no.~5, pp. 1770--1779, 2010.

\bibitem{huang2011frequency}
X.~Huang and Y.~J. Guo, ``Frequency-domain {AoA} estimation and beamforming
  with wideband hybrid arrays,'' \emph{IEEE Trans. Wireless Commun.}, vol.~10,
  no.~8, pp. 2543--2553, 2011.

\bibitem{singh2015feasibility}
J.~Singh and S.~Ramakrishna, ``On the feasibility of codebook-based beamforming
  in millimeter wave systems with multiple antenna arrays,'' \emph{IEEE Trans.
  Wireless Commun.}, vol.~14, no.~5, pp. 2670--2683, 2015.

\bibitem{noh2017multiresolutionCodebook}
S.~Noh, M.~D. Zoltowski, and D.~J. Love, ``Multi-resolution codebook and
  adaptive beamforming sequence design for millimeter wave beam alignment,''
  \emph{IEEE Trans. Wireless Commun.}, vol.~16, no.~9, pp. 5689--5701, 2017.

\bibitem{lee2014exploitingCS}
J.~Lee, G.-T. Gil, and Y.~H. Lee, ``Exploiting spatial sparsity for estimating
  channels of hybrid {MIMO} systems in millimeter wave communications,'' in
  \emph{Proc. IEEE Globecom}.\hskip 1em plus 0.5em minus 0.4em\relax IEEE,
  2014, pp. 3326--3331.

\bibitem{chuang2015high}
S.-F. Chuang, W.-R. Wu, and Y.-T. Liu, ``High-resolution {AoA} estimation for
  hybrid antenna arrays,'' \emph{IEEE Trans. Antennas Propagat.}, vol.~63,
  no.~7, pp. 2955--2968, 2015.

\bibitem{lu2024integrated}
S.~Lu \emph{et~al.}, ``Integrated sensing and communications: Recent advances
  and ten open challenges,'' \emph{IEEE Internet Things J.}, 2024.

\bibitem{wu2024sensing}
K.~Wu \emph{et~al.}, ``Sensing in bi-static {ISAC} systems with clock
  asynchronism: A signal processing perspective,'' \emph{arXiv preprint
  arXiv:2402.09048}, 2024.

\bibitem{li2017indotrack}
X.~Li \emph{et~al.}, ``Indotrack: Device-free indoor human tracking with
  commodity {Wi-Fi},'' \emph{Proc. ACM Interact. Mob. Wearable Ubiquitous
  Technol.}, vol.~1, no.~3, pp. 1--22, 2017.

\bibitem{qian2018widar2}
K.~Qian, C.~Wu, Y.~Zhang, G.~Zhang, Z.~Yang, and Y.~Liu, ``Widar2. 0: Passive
  human tracking with a single {Wi-Fi} link,'' in \emph{Proc. 16th Annu. Int.
  Conf. Mobile Syst., Appl., Serv.}, 2018, pp. 350--361.

\bibitem{zeng2019farsense}
Y.~Zeng, D.~Wu, J.~Xiong, E.~Yi, R.~Gao, and D.~Zhang, ``Farsense: Pushing the
  range limit of {WiFi}-based respiration sensing with {CSI} ratio of two
  antennas,'' \emph{Proc. ACM Interact. Mob. Wearable Ubiquitous Technol.},
  vol.~3, no.~3, pp. 1--26, 2019.

\bibitem{li2022csi}
X.~Li, J.~A. Zhang, K.~Wu, Y.~Cui, and X.~Jing, ``{CSI}-ratio-based {Doppler}
  frequency estimation in integrated sensing and communications,'' \emph{IEEE
  Sensors J.}, vol.~22, no.~21, pp. 20\,886--20\,895, 2022.

\bibitem{ni2021uplink-CACC}
Z.~Ni, J.~A. Zhang, X.~Huang, K.~Yang, and J.~Yuan, ``Uplink sensing in
  perceptive mobile networks with asynchronous transceivers,'' \emph{IEEE
  Trans. Signal Process.}, vol.~69, pp. 1287--1300, 2021.

\bibitem{ni2023uplink-CSIR}
Z.~Ni, J.~A. Zhang, K.~Wu, and R.~P. Liu, ``Uplink sensing using {CSI} ratio in
  perceptive mobile networks,'' \emph{IEEE Trans. Signal Process.}, 2023.

\bibitem{kotaru2015spotfi}
M.~Kotaru, K.~Joshi, D.~Bharadia, and S.~Katti, ``Spotfi: Decimeter level
  localization using {WiFi},'' in \emph{Proc. ACM Conf. Special Interest Group
  Data Commun}, 2015, pp. 269--282.

\bibitem{li2016dynamic-MUSIC}
X.~Li, S.~Li, D.~Zhang, J.~Xiong, Y.~Wang, and H.~Mei, ``Dynamic-music:
  Accurate device-free indoor localization,'' in \emph{Proc. ACM Int. Joint
  Conf. Pervasive Ubiquitous Comput.}, 2016, pp. 196--207.

\bibitem{pegoraro2024jump}
J.~Pegoraro \emph{et~al.}, ``Jump: Joint communication and sensing with
  unsynchronized transceivers made practical,'' \emph{IEEE Trans. Wireless
  Commun.}, 2024.

\bibitem{zhao2023multiple}
J.~Zhao, Z.~Lu, J.~A. Zhang, S.~Dong, and S.~Zhou, ``Multiple-target {Doppler}
  frequency estimation in {ISAC} with clock asynchronism,'' \emph{IEEE Trans.
  Veh. Technol.}, 2023.

\bibitem{heath2016overview}
R.~W. Heath, N.~Gonzalez-Prelcic, S.~Rangan, W.~Roh, and A.~M. Sayeed, ``An
  overview of signal processing techniques for millimeter wave {MIMO}
  systems,'' \emph{IEEE J. Sel. Top. Signal Process.}, vol.~10, no.~3, pp.
  436--453, 2016.

\bibitem{giordani2018tutorial}
M.~Giordani, M.~Polese, A.~Roy, D.~Castor, and M.~Zorzi, ``A tutorial on beam
  management for {3GPP} {NR} at {mmWave} frequencies,'' \emph{IEEE Commun.
  Surv. Tutorials}, vol.~21, no.~1, pp. 173--196, 2018.

\bibitem{sadhu202224}
B.~Sadhu \emph{et~al.}, ``A 24--30-{GHz} 256-element dual-polarized {5G} phased
  array using fast on-chip beam calculators and magnetoelectric dipole
  antennas,'' \emph{IEEE J. Solid-State Circuits}, vol.~57, no.~12, pp.
  3599--3616, 2022.

\bibitem{shan1985spatial}
T.-J. Shan, M.~Wax, and T.~Kailath, ``On spatial smoothing for
  direction-of-arrival estimation of coherent signals,'' \emph{IEEE Trans.
  Acoust., Speech, Signal Process.}, vol.~33, no.~4, pp. 806--811, 1985.

\bibitem{zhang2018multibeam}
J.~A. Zhang, X.~Huang, Y.~J. Guo, J.~Yuan, and R.~W. Heath, ``Multibeam for
  joint communication and radar sensing using steerable analog antenna
  arrays,'' \emph{IEEE Trans. Veh. Technol.}, vol.~68, no.~1, pp. 671--685,
  2018.

\end{thebibliography}
 
\end{document}